\documentclass[onecolumn,12pt]{IEEEtran}
\usepackage{amsmath}
\usepackage{latexsym}
\usepackage{amssymb}
\usepackage{bm}
  \usepackage[pdftex]{graphicx}
\newtheorem{thm}    {Theorem}
\newtheorem{lem}     {Lemma}
\newtheorem{corollary}  {Corollary}
\newtheorem{proposition}        {Proposition}
\newtheorem{define}     {Definition}
\def\argmax{\mathop{\rm argmax}}
\def\uni{\mathop{\rm uni}}

\def\cX{{\cal X}}

\def\rE{{\rm E}}
\def\Pr{{\rm Pr}}

\newcommand{\bF}{\mathbb{F}}
\newcommand{\sM}{\mathsf{M}}
\newcommand{\sL}{\mathsf{L}}

\def\QED{\mbox{\rule[0pt]{1.5ex}{1.5ex}}}

\def\endproof{\hspace*{\fill}~\QED\par\endtrivlist\unskip}

\def\Label#1{\label{#1}\ [\ #1\ ]\ }
\def\Label{\label}

\begin{document}
\title{Secure list decoding
\thanks{The material in this paper was presented in part at the 
IEEE International Symposium on Information Theory (ISIT2019),
Paris, France, July 7 -- 12, 2019 \cite{MH17-1}.}
}
\author{
Masahito~Hayashi~\IEEEmembership{Fellow,~IEEE}\thanks{Masahito Hayashi is with 
Shenzhen Institute for Quantum Science and Engineering, Southern University of Science and Technology,
Shenzhen, 518055, China,
Graduate School of Mathematics, Nagoya University, Nagoya, 464-8602, Japan,
Center for Quantum Computing, Peng Cheng Laboratory, Shenzhen 518000, China,
and the Centre for Quantum Technologies, National University of Singapore, 3 Science Drive 2, 117542, Singapore
(e-mail:hayashi@sustech.edu.cn, masahito@math.nagoya-u.ac.jp)}}
\date{}
\maketitle
\begin{abstract}
We propose a new concept of secure list decoding.
While the conventional list decoding requires that the list contains the transmitted message,
secure list decoding requires the following additional security conditions.
The first additional security condition is the impossibility of the correct decoding, i.e.,
the receiver cannot uniquely identify the transmitted message
even though the transmitted message is contained in the list. 
This condition can be trivially satisfied when the transmission rate is larger 
than the channel capacity.
The other additional security condition is the impossibility for the sender to estimate another element of the decoded list except for the transmitted message.
This protocol can be used for anonymous auction, which realizes
the anonymity for bidding.
\end{abstract}
\begin{IEEEkeywords}
list decoding; anonymous auction;
security condition; capacity region
\end{IEEEkeywords}

\section{Introduction}\Label{S1}
Relaxing the condition of the decoding process,
Elias \cite{Elias} and Wozencraft \cite{Wo} independently
introduced list decoding as the method to allow more than one element
as candidates of the message sent by the encoder at the decoder.
When one of these elements coincides with the true message, the decoding is regarded as successful.
The paper \cite{Guruswami}
discussed its algorithmic aspect.
In this formulation, Nishimura \cite{Nish} obtained 
the channel capacity by showing its strong converse part\footnote{The strong converse part is the argument that the average error goes to $1$ if the code has a transmission rate over the capacity.}.
That is, he showed that the transmission rate is less than the conventional capacity plus 
the rate of the list size, i.e., the number of list elements.
Then, the reliable transmission rate does not increase
even when list decoding is allowed
if the list size does not increase exponentially.
In the non-exponential case, these results were generalized by Ahlswede \cite{Ah}.
Further, the paper \cite{Haya} showed that 
the upper bound of capacity by Nishimura can be attained even if the list size increases exponentially.
When the number of lists is $\sL$,
the capacity can be achieved by choosing the same codeword for $\sL$ distinct messages.

However, the merit of increase in the list size
was not discussed sufficiently.
To get a merit of list coding, we need a code construction that is essentially different from the conventional coding.
Since the above capacity-achieving code construction does not have an essential difference from the conventional coding,
we need to rule out the above type of construction of list coding.
That is, to extract a merit of list decoding, we need additional parameters to characterize
the difference from the conventional code construction, which can be expected to rule out such a trivial construction.

In this paper, we propose a new concept of secure list decoding
by adding several security conditions, which can be considered as additional constraints.
To explain this protocol, we consider the following anonymous auction scenario, which realizes
the anonymity for bidding.
$\sM$ players participate in the auction for an item dealt by Bob,
and they have their distinct ID from $1$ to $\sM$.
\begin{description}
\item[(i)] (Bidding) Several players bid with their ID.
Each of them sends his/her ID via a noisy channel with his/her price.
Then, for each bid,
the dealer, Bob outputs $\sL$ ID numbers as the list. 
The list is required to contain the ID of the player making the bid.

\item[(ii)] (Purchasing) 
Assume that a player, Alice has an ID $M$ and her bidding price is highest.
She purchases the item from Bob by showing her ID $M$.

\end{description}
This scenario has the following requirements.
\begin{description}
\item[(a)] Bob wants to identify whether 
the person to purchase the item is the same as the person to bid the highest price.
That is, $M$ needs to be one of 
$\sL$ ID numbers $M_1, \ldots, M_{\sL} $ output by Bob in her bid.
\item[(b)] Alice wants to hide her ID $M$ at the bidding step (i). 
Hence, she will not be identified by Bob when she loses this auction.
\item[(c)] Bob wants to avoid the situation that two players show him the correct ID 
at purchasing Step (ii). That is, Alice cannot find another element among 
$\sL$ ID numbers $M_1, \ldots, M_{\sL} $ output by Bob in her bid
except for $M$.
\end{description}
The requirement (a) is the condition for the requirement for the conventional list decoding
while the requirements (b) and (c) are not considered in the conventional list decoding.
In this paper, as a new concept to satisfy these conditions, we propose secure list decoding
by imposing the following two additional conditions to the list decoding.
The first additional security condition is the impossibility of the correct decoding.
That is, the receiver cannot uniquely identify the transmitted message
even though the transmitted message is contained in the list. 
This condition can be trivially satisfied when the transmission rate is larger than the channel capacity
due to the strong converse property
under the asymptotic setting with the discrete memory less channel.
The other additional security condition is the impossibility for the sender to estimate another element of the decoded list except for the transmitted message.
In fact, we might use an authentication protocol to identify Alice \cite{KR94}.
In this case, if Alice gives the key for the authentication to the third party,
the third party can claim to Bob that he is also the winner of this auction.
To avoid this type of spoofing, we need to use the ID number. 
That is, the above anonymous auction scenario realizes a kind of authentication, 
which satisfies the anonymity
and forbids spoofing even when Alice colludes the third party.
In this paper, we formulate secure list decoding,
and define various types of capacity regions for secure list decoding
under the asymptotic setting with the discrete memory less channel.
Then, we calculate these capacity regions under several conditions.

This paper is structured as follows. 
Section \ref{S2-1} gives the formulation of secure list decoding.
Section \ref{BC} explains the relation with bit commitment.
Section \ref{S51} prepares several information quantities.
Section \ref{S51} states the main result by deriving the capacity regions.
Section \ref{S-C} shows the converse part, and 
Section \ref{S-K} proves the direct part.

\section{Problem setting}\Label{S2}
\subsection{Our setting with intuitive description}\Label{S2-1}
To realize the requirements (a), (b), and (c) mentioned in Section \ref{S1}, 
we formulate the mathematical conditions for the protocol for 
a given channel $W$ from the discrete system ${\cal X}$ to the other system ${\cal Y}$
with integers $\sL < \sM$ and security parameters $\epsilon_A,\delta_B,\delta_C$.
In the following, we describe the condition in an intuitive form in the first step. Later, we 
transform it into a coding-theoretic form
because the coding-theoretic form matches the theoretical discussion including the proofs of our main results.

Alice sends her ID $M \in {\cal M}:= \{1, \ldots, \sM\}$ via a noisy channel 
with a code $\phi$, which is a map from ${\cal M} $ to ${\cal X}$.
Bob outputs the $\sL$ messages $M_1, \ldots M_{\sL}$.
The decoder is given as the following $\Psi$.
For $y\in {\cal Y}$, we choose a subset $\Psi(y) \subset {\cal M}$ with $|\Psi(y)|= \sL$.

Then, we impose the following conditions
for an encoder $\phi$ and a decoder $\Psi$.
\begin{description}
\item[(A)] 
Verifiable condition. Any element $m \in {\cal M}$ satisfies
\begin{align}
    \Pr [m \notin \Psi(Y) | X = \phi(m)] \le \epsilon_A.
    \end{align}
\item[(B)]
Non-decodable condition.
Any single-element decoder $\psi: {\cal Y} \to {\cal M}$ satisfies
\begin{align}
   \frac{1}{\sM} \sum_m \Pr [ \psi(Y) = m | X = \phi(m)] \le \delta_B .
\end{align}
\item[(C)]
Non-cheating condition for honest Alice.
Any distinct pair $m'\neq m$ satisfies
\begin{align}
\Pr [m' \in \Psi(Y) | X=\phi(m)] \le \delta_C . 
\end{align}
\end{description}

Now, we discuss how the code $(\phi,\Psi)$ can be used for the task explained in Section \ref{S1}.
Assume that Alice sends her ID $M$ to Bob 
by using the encoder $\phi$ via noisy channel $W$
and Bob gets the list $M_1, \ldots, M_{\sL}$ by applying the decoder $\Psi$
at Step (i).
At Step (ii), Alice shows her ID $M$ to Bob.
Verifiable condition (A) guarantees that her ID $M$ belongs to Bob's list. 
Hence, the requirement (a) is satisfied.
Non-decodable condition (B) forbids Bob to identify Alice's ID at Step (i), hence
it guarantees the requirement (b).
In fact, if 
$m$ is Alice's ID and there exist another ID $m' (\neq m)$ and 
an element  $x_0 $ that might be different from $\phi(m)$
such that
$\Pr [m \in \Psi(Y) | X=x_0] $ and 
$\Pr [m' \in \Psi(Y) | X=x_0] $
are close to $1$,
Alice can make the following cheating by sending $x_0$ instead of $\phi(m)$. 
Since Alice knows that $m'$ belongs to Bob's decoded list,
she finds the third person whose ID is $m'$.
Then, she tells the third person this fact.
At Step (ii), the third person can make spoofing by showing Bob his/her ID. 
Since Non-cheating condition (C) forbids Alice such a cheating, 
it guarantees the requirement (c).
Further, Bob is allowed to decode less than $\sL$ messages.
That is, $\sL$ is the maximum number that Bob can list as the candidates of the original message.

However, Condition (C) is the security evaluation for honest Alice who uses the correct encoder $\phi$.
Dishonest Alice might send her message by using a different encoder.
To cover such a case, we impose the following condition instead of Condition (C).

\begin{description}
\item[(D)] Non-cheating condition for dishonest Alice.
If a pair of $x$ and $m$ satisfies
\begin{align}
   \Pr [m \notin \Psi(Y) | X=x] \le \frac{1}{2}, \hbox{ i.e., }
   \Pr [m \in \Psi(Y) | X=x] > \frac{1}{2},
\Label{HB2}
\end{align}
any $m'(\neq m)$ satisfies
\begin{align}
\Pr [m' \in \Psi(Y) | X=x] \le \delta_C . 
\Label{HB3}
\end{align}
\end{description}
In the following, 
when a code $(\phi,\Psi) $ satisfies conditions (A), (B) and (D), 
it is called an $(\epsilon_A,\delta_B,\delta_C)$ code.

Now, we observe how to characterize the code constructed to achieve the capacity in the paper \cite{Haya}.
For this characterization, we consider the following code when $\sM' \sL=\sM$.
We divide the $\sM$ messages into $\sM'$ groups whose group is composed of $\sL$ messages.
First, we prepare a code $({\phi}',{\psi}')$ to transmit the message with size $\sM'$ with decoding error probability $\epsilon_A'$,
where ${\phi}'$ is the encoder and ${\psi}'$ is the decoder.
When the message $M$ belongs to the $i$-th group,
Alice sends ${\phi}'(i)$.
Using the decoder ${\psi}'$, Bob recovers $i'$.
Then, Bob outputs $\sL$ elements that belongs to the $i'$-th group.
In this code, the parameter $\delta_B$ is given as $1/\sL$. Hence, it satisfies the non-decodable condition with a good parameter.
However, the parameter $\delta_C$ becomes at least $1-\epsilon_A'$.
Hence, this protocol essentially does not satisfy Non-cheating condition (C) nor (D).
In this way, our security parameter rules out the above trivial code construction.

\subsection{Relation to bit commitment}\Label{BC}
If our task is realized and ${\cal M}$ is a vector space $\bF_2^t$ over the finite field $\bF_2$, we can approximately realize 
bit commitment as follows while it is known that bit commitment can be realized by using noisy channel
\cite{BC1,BC2,BC3}.
To explain this realization, we state the formulation of approximate bit commitment as follows.
\begin{description}
\item[(1)] Committing phase: Alice has bit $X= 0$ or $1$. She makes a commitment to Bob.
\item[(2)] Revealing phase: Alice announces her bit to Bob. 
Bob checks whether it is true.
\end{description}
There are two kinds of cheating.
One is Bob's cheating. Bob guesses Alice's bit only from the result of the committing phase.
The other is Alice's cheating. 
In the revealing phase, Alice convinces Bob that Alice's bit is $X \oplus 1$.
Then, we impose three conditions.
\begin{description}
\item[(B1)] The mutual information between Alice's bit $X$ and Bob's result of the committing phase
is upper bounded by $\bar{\delta}_B$.
\item[(B2)] The successful maximum probability of Alice's cheating under the following condition
(B2-1) is upper bounded by $\bar{\delta}_A$.
\begin{description}
\item[(B2-1)]When Alice is honest in the revealing phase,
the probability that Alice convinces Bob that Alice's bit is $X$ is at least $1/2$.
\end{description}
 \item[(B3)] Another parameter is the error probability.
When Alice and Bob are honest, the probability that
Bob is convinced that Alice's bit is lower bounded by $X$ by $1-\bar{\epsilon}_A$.
\end{description}
If a protocol satisfies the above property, 
the protocol is is called a $(\bar{\epsilon}_A,\bar{\delta}_A,\bar{\delta}_B)$ approximate bit commitment.
When $\bar{\delta}_A=\bar{\delta}_B=0$, it is called perfect bit commitment.

Given a $(\epsilon_A,\delta_B,\delta_C)$ code $(\phi,\Psi) $, 
we construct a protocol for an approximate bit commitment as follows.
Assume that the message set of the code $(\phi,\Psi) $ is ${\cal M}=\bF_2^t$.
Then, $X$, $M$, and $Y$ are variables given in Section \ref{S2-1} with the code $(\phi,D) $.
Hence, $M$ is subject to the uniform distribution.
Define the conditional Renyi entropy of order 2 
by $H_2(M|Y):= -\log \sum_{y,m} P_Y(y) P_{M|Y}(m|y)^2 $.
We have $H_2(M|Y)\ge - \log \delta_B$ \cite[Theorem 1]{KRS}\cite[Lemma 5.9 and (5.134)]{Hay17}.
Due to the universal hashing lemma\cite{BBCM,HILL}
\cite[Theorem 2]{H-leaked}, there exists 
a surjective homomorphic\footnote{For a surjective homomorphic
unversal hash function, see \cite[Section II]{Tsuru}.} 
hash function $f$ from ${\cal M}$ to $\bF_2$ 
such that
\begin{align}
I(f(M);Y) \le 2 e^{-H_2(M|Y)} \le 2 \delta_B .
\Label{H20}
\end{align}
Using the above hash function $f$, we describe our protocol for approximate bit commitment as follows.

\begin{description}
\item[(1)] Commiting phase: Alice has bit $X= 0$ or $1$. 
She randomly choose one element $M$ from $f^{-1}(X)$.
She sends $M'$ by using the encoder $\phi$.
Bob output $\sL$ elements $M_1, \ldots, M_{\sL}$ from $Y$.

\item[(2)] Revealing phase: To convince Bob that her bit is $X$,
Alice announces $M'$ to Bob. 
Bob checks whether $f(M')=X$ and $M'$ is one element of 
$M_1, \ldots, M_{\sL}$.
\end{description}

\begin{thm}
The above protocol is 
an $(\epsilon_A,{\delta}_C, 2{\delta}_B)$ approximate bit commitment.
\hfill $\square$\end{thm}
When the parameters $\epsilon_A,\delta_B$, and $\delta_C$ are sufficiently small, 
the above protocol works as bit commitment due to this theorem.

\begin{proof}
Since the hash function $f$ is surjective homomorphic,
$|f^{-1}(0)|=|f^{-1}(1)|=2^{t-1}$.
Hence, the random variable $M'$ chosen in the committing phase
is subject to the uniform distribution.
Hence, $M'$ can be considered as the variable $M$ because they are subject to the same distribution.
Since Bob's list contains $M'$ at least probability $1-\epsilon_A$, 
the condition (B3) holds.
Also, \eqref{H20} implies the condition (B1).

Assume that Bob outputs $\sL$ elements $M_1, \ldots, M_{\sL}$.
To satisfy the condition (B2-1), 
Alice needs to prepare the pair of $x$ and $m$ to satisfy the condition \eqref{HB2}.
That is, Alice needs to send an element $x$ to satisfy the condition \eqref{HB2} 
with keeping an element $m \in f^{-1}(X)$. 
To convince Bob that her bit is $X\oplus 1$,
in the revealing phase,
Alice sends Bob an element in $f^{-1}(X \oplus 1) $
that needs to belong to the set $\{M_1, \ldots, M_{\sL}\}$.
However, due to \eqref{HB3},
any element in $f^{-1}(X \oplus 1) $ satisfies this condition at most 
probability $\delta_C$.
Thus, we obtain the condition (B2).
\end{proof}

\subsection{Our setting with coding-theoretic description}\Label{S2-1B}
To rewrite the above conditions in a coding-theoretic way, 
we introduce several notations.
For $x\in {\cal X}$ and a distribution on ${\cal X}$, we define the distribution $W_x$ and $W_{P}$ on ${\cal Y}$
as $W_x(y):=W(y|x)$ and $W_P(y):= \sum_{x\in {\cal X}}P(x)W(y|x)$.
Alice sends her ID $M \in {\cal M}:= \{1, \ldots, \sM\}$ via noisy channel 
$W$.
Bob' decoder 
can be described as disjoint subsets 
$D=\{{\cal D}_{m_1, \ldots, m_{\sL}}\}_{\{m_1, \ldots, m_{\sL}\} \subset {\cal M}}$
such that $\cup_{\{m_1, \ldots, m_{\sL}\} \subset {\cal M}} {\cal D}_{m_1, \ldots, m_{\sL}}={\cal Y}$.
That is, 
we have the relation $
{\cal D}_{m_1, \ldots, m_{\sL}}=\{y| \{m_1, \ldots, m_{\sL}\}= \Psi(y)\}$.
In the following, we denote our decoder by $D$ instead of $\Psi$.

In particular, 
when a decoder 
has only one outcome as an element of ${\cal M}$ 
it is called a single-element decoder.
It is given as disjoint subsets 
$\tilde{{\cal D}}=\{ \tilde{{\cal D}}_{m}\}_{m \in {\cal M}}$
such that $\cup_{m \in {\cal M} } \tilde{{\cal D}}_{m}={\cal Y}$.

\begin{thm}
The conditions (A) -- (D) for an encoder $\phi$ and a decoder $D=\{{\cal D}_{m_1, \ldots, m_{\sL}}\}_{\{m_1, \ldots, m_{\sL}\} \subset {\cal M}}$
are rewritten in a coding-theoretic way as follows.
\begin{description}
\item[(A)] 
Verifiable condition.
\begin{align}
\epsilon_{A}(\phi,D)
&:=\max_{m \in {\cal M}}\epsilon_{A,m}(\phi(m),D) 
\le  \epsilon_A \\
\epsilon_{A,m}(x,D) 
&:=
1-\sum_{m_1, \ldots, m_{\sL}:\{m_1, \ldots, m_{\sL}\} \ni m }
W_{x}({\cal D}_{m_1, \ldots, m_{\sL}}).  
\end{align}

\item[(B)]
Non-decodable condition.
\begin{align}
\delta_{B}(\phi) 
&:=\max_{\tilde{D}=\{ \tilde{{\cal D}}_{m}\}_{m \in {\cal M}}}
\sum_{m=1}^{\sM}\frac{1}{\sM}
\delta_{B,\phi(m)}(\tilde{{\cal D}}_{m}) 
\le  \delta_B ,
\end{align}
where the above maximum is taken for 
a single-element decoder $\tilde{D}=\{ \tilde{{\cal D}}_{m}\}_{m \in {\cal M}}$,
and $\delta_{B,\phi(m)}(\tilde{{\cal D}}_{m}) $
is defined for a single-element decoder $\tilde{D}$ as
\begin{align}
\delta_{B,x}(\tilde{{\cal D}}_{m}) 
&:=W_{x} (\tilde{{\cal D}}_{m}  ) .
\end{align}
\item[(C)]
Non-cheating condition for honest Alice.
\begin{align}
\delta_{C}(\phi,D)
&:=
\max_{m \in {\cal M}}
\delta_{C,m}(\phi(m),D) \le  \delta_C 
\\
\delta_{C,m}(x,D)  
&:=\max_{m' (\neq m) \in {\cal M}}
\sum_{ m_1, \ldots, m_{\sL}:\{m_1, \ldots, m_{\sL}\} \ni m' }
W_{x}({\cal D}_{m_1, \ldots, m_{\sL}}) .
\Label{EE8}
\end{align}
\item[(D)] Non-cheating condition for dishonest Alice.
\begin{align}
\delta_{D}(D)&:=
\max_{m \in {\cal M}}
\delta_{D,m}(D)
\le \delta_C \\
\delta_{D,m}(D)&:=
\max_{x \in {\cal X}}
\Big\{\delta_{C,m}(x,{D}) \Big|
\epsilon_{A,m}(x,D)
\le  \frac{1}{2}
\Big\}.
\end{align}
\end{description}
\hfill $\square$\end{thm}

\begin{proof}
For any $m\in {\cal M}$ and $y \in {\cal Y}$,
the condition $m \in \Psi(y)$ is equivalent to the condition $y \in \cup _{ m_1, \ldots, m_{\sL}:\{m_1, \ldots, m_{\sL}\} \ni m'} {\cal D}_{m_1, \ldots, m_{\sL}}$.
Since $$
\sum_{ m_1, \ldots, m_{\sL}:\{m_1, \ldots, m_{\sL}\} \ni m }
W_{x}({\cal D}_{m_1, \ldots, m_{\sL}}) 
=W_{x}\Big(\bigcup _{ m_1, \ldots, m_{\sL}:\{m_1, \ldots, m_{\sL}\} \ni m} {\cal D}_{m_1, \ldots, m_{\sL}}
\Big),
$$ we obtain the equivalence between the conditions (A) and (C) given in Section \ref{S2-1} and those given here. 
As the condition $m=\psi(y)$ is equivalent to the condition $ y \in \tilde{{\cal D}}_m$,
we obtain the desired equivalence for the condition (B).
Also, the condition \eqref{HB2} is equivalent to the condition 
$\epsilon_{A,m}(x,D) \le  \frac{1}{2}$, which implies the desired equivalence with respect to the condition (D).
\end{proof}

In the following, 
when a code $(\phi,D) $ satisfies conditions (A), (B) and (D), 
it is called an $(\epsilon_A,\delta_B,\delta_C)$ code.
Also, for a code $(\phi,D) $,
we denote $\sM$ and $\sL$
by $|(\phi,D)|_1$ and $|(\phi,D)|_2$.
Also, we allow stochastic encoder, in which $\phi(m)$ is a distribution on ${\cal X}$.
In this case, for a function $f$ from ${\cal X}$ to $\mathbb{R}$,
$f(\phi(m))$ expresses $\sum_{x}f(x)\phi(m)(x)$.

\section{Information quantities}\Label{S51}
\subsection{Notation based on distribution and conditional distribution}\Label{S511}
Let ${\cal X}$ be a finite set and denote the set of probability distributions on ${\cal X}$ by ${\cal P}({\cal X})$.
Consider the channel written as the transition matrix 
$W$ from ${\cal X}$ to ${\cal Y}$.
For $x\in {\cal X}$ and a distribution $P \in {\cal P}({\cal X})$, we define the distribution $W_x$ and $W_{P}$ on ${\cal Y}$
as $W_x(y):=W(y|x)$ and $W_P(y):= \sum_{x\in {\cal X}}P(x)W(y|x)$.
We assume that $W_x\neq W_{x'}$ for $x\neq x' \in\cX$.
In the following,
$\mathbb{E}_x$ expresses the average 
with respect to a variable over the system ${\cal Y}$
under the distribution $W_x$
and
$\mathbb{V}_x$ expresses the variance with respect to a variable over the system ${\cal Y}$
under the distribution $W_x$.
This notation is also applied to the $n$-fold extended setting.

We define
\begin{align}
C(W)&:=\max_{P \in {\cal P}({\cal X})}I(P,W),\Label{CU} \\
I(P,W)&:=
\sum_{x\in {\cal X}}P(x) \sum_{y \in {\cal Y}} W_x(y)
(\log W_x(y)-\log W_P(y)) ,\\
H(P)& :=-\sum_{x\in {\cal X}}P(x) \log P(x), 
\end{align}
where the base of logarithm is $2$.

For $x,x' \in {\cal X}$, we define
\begin{align}
F(x,x'|P):=
\mathbb{E}_x (\log W_{x'}(Y)-\log W_P(Y))
=D(W_x\|W_P)-D(W_x\|W_{x'}).
\end{align}
Then, we define
\begin{align}
\zeta_2(P) &:= \max_{x \neq x'\in {\cal X}} \max_{x''\in {\cal X}} F(x,x''|P)-F(x',x''|P)  \Label{zetM}\\
\zeta_1(P) &:= \min_{x \neq x'\in {\cal X}} F(x,x|P)-F(x',x|P) 
=\min_{x \neq x'}
D(W_{x'}\|W_{x})+D(W_x\|W_P) -D(W_{x'}\|W_P)
. \Label{zetm}
\end{align}
In this paper, the condition
\begin{align}
\zeta_1(P)>0\Label{IMC}
\end{align}
plays an important role.
\begin{lem}\Label{LX1}
When $P \in {\cal P}({\cal X})$ satisfies the condition $I(P,W)=C(W)$, the condition \eqref{IMC} holds.
\end{lem}

\begin{proof}
Since $D(W_x\|W_P) =D(W_{x'}\|W_P)$,
we have
\begin{align}
\zeta_1(P_0) =\min_{x \neq x'}
D(W_{x'}\|W_{x})>0.
\end{align}
\end{proof}

Also, we often impose the following condition for our channel $W$;
\begin{align}
V(W):=\max_{x,x'\in {\cal X}}\mathbb{V}_x (\log W_{x'}(Y)-\log W_P(Y)) < \infty.\Label{OfIm}
\end{align}

\subsection{Notation based on variables}
When we focus on a Markov chain $U-X-Y$ with a variable on a finite set ${\cal U}$,
it is intuitive to handle notations based on the variables $U,X, $ and $Y$.
In this paper, the conditional distribution on $Y$ conditioned with $X$ is fixed to the channel $W$.
It is sufficient to fix a joint distribution $P \in {\cal P}({\cal U}\times {\cal X})$.
To clarify this dependence, we add the subscript $~_P$ as $H(X)_P,I(X;Y)_P,$ and $H(X|U)_P$, etc.
In fact, the notation given in Section \ref{S511} will be used for our proof of the direct part,
the notation given in this section will be used for the characterization of our rate region and 
our proof of the converse part.
This is because the characterization of our rate region and 
the converse part mainly discuss Markovian chains
while the direct part mainly evaluates  the parameters $\epsilon_A,\delta_B,$ and $\delta_D$.

Then, we define 
\begin{align}
H_0 &:=\max_{P \in {\cal P}({\cal X})} \{H(X)_P| I(X;Y)_P=C(W)\}, \\
P_{\max} &:=\argmax_{P \in {\cal P}({\cal X})} \{H(X)_P| I(X;Y)_P=C(W)\},
\end{align}
and the function $\kappa$ for $R_1 \ge H_0$ as
\begin{align}
\kappa(R_1):= 
\left \{
\begin{array}{ll}
\max_{P\in {\cal P}({\cal U}\times {\cal X})} \{ H(X|Y,U)_P |  H(X|U)_P= R_1 \}  & \hbox{ when } R_1 > H_0\\
\protect{[} R_1 - C(W) ]_+ & \hbox{ when } R_1 \le H_0.
\end{array}
\right .
\end{align}
In the above definition, the size of the finite set ${\cal U}$ can be chosen to be arbitrarily large.
When a maximization $\max$ or a union $\cup$ with respect to $P\in {\cal P}({\cal U}\times {\cal X})$ appears
 in the remaining part,
this rule for the finite set ${\cal U}$ is applied.
Then, we define the set ${\cal P}_0$ as
\begin{align}
{\cal P}_0:=\{P\in {\cal P}({\cal X})| \kappa( H(X)_P)=H(X|Y)_P\}.
\end{align}
For example, the distribution $P_{\max}$ and the uniform distribution $P_{\uni}$ on ${\cal X}$ belong to ${\cal P}_0$.
\begin{lem}
We have
\begin{align}
\{P\in {\cal P}_0 | H(X)_P \le H_0\}
=\{P\in {\cal P}({\cal X})| I(X;Y)_P=C(W)\}.
\end{align}
\end{lem}
\begin{proof}
When $H(X)_P \le H_0$, 
we have $H(X|Y)_P =H(X)_P-I(X;Y)_P= [H(X)_P-I(X;Y)_P]_+ \ge [H(X)_P-C(W)]_+=\kappa (H(X)_P)$.
Also, the condition $P\in {\cal P}_0$ implies the condition $[H(X)_P-C(W)]_+=H(X|Y)_P$.
The combination of the above conditions implies the condition $I(X;Y)_P=C(W)$.

Conversely, 
the condition $I(X;Y)_P=C(W)$ implies the conditions $H(X)_P \le H_0$.
and $H(X|Y)_P = [H(X)_P-C(W)]_+=\kappa (H(X)_P)$.
Hence, the desired relation is obtained.
\end{proof}

Then, we prepare the following lemmas whose proofs are given in Appendices \ref{A01} and \ref{A02}.
\begin{lem}\Label{LA1}
Given a joint distribution $P\in {\cal P}({\cal U}\times {\cal X})$, we have the Markov chain $U-X-Y$, and focus on  the information quantities
$ I(X;Y|U)_P$ and $H(X|U)_P$.
Then, we have
\begin{align}
{\cal C}
:=&
\cup_{P\in {\cal P}({\cal U}\times {\cal X})} \{ (R_1,R_2) | 0 < R_1-R_2< I(X;Y|U)_P , ~R_1 < H(X|U)_P ,~0< R_1,~0< R_2\}  \nonumber \\
=&
\{ (R_1,R_2) | 0 < R_1< \log |{\cal X}|, ~\kappa(R_1) < R_2 < R_1 \}  .
 \end{align}
\hfill $\square$
\end{lem}

\begin{lem}\Label{LA2}
We have
\begin{align}
&{\cal C}\cap \{(R_1,R_2)| H_0 \le R_1 \le \log |{\cal X}| \}\nonumber \\
=&\mathcal{CO}(
\cup_{P \in {\cal P}_0}
\{(R_1,R_2)| 
H_0 < R_1 < H(X)_P ,~
R_1- I(X;Y)_P< R_2< R_1\}),
\end{align}
where $\mathcal{CO}$ expresses the convex full.
\hfill $\square$
\end{lem}

\section{Main results}\Label{S51}
To give the capacity region, 
we consider $n$-fold discrete memoryless extension $W^n$ of the channel $W$.
A sequence of codes $\{(\phi_n,D_n)\}$ is called strongly secure 
when 
$\epsilon_A(\phi_n,D_n) \to 0$,
$\delta_B(\phi_n) \to 0$,
$\delta_D(D_n) \to 0$.
A sequence of codes $\{(\phi_n,D_n)\}$ is called weakly secure 
when 
$\epsilon_A(\phi_n,D_n) \to 0$,
$\delta_B(\phi_n) \to 0$,
$\delta_C(\phi_n,D_n) \to 0$.
A rate pair $(R_1,R_2)$ is strongly deterministically (stochastically) achievable when
there exists a strongly secure sequence of deterministic (stochastic) codes $\{(\phi_n,D_n)\}$
such that 
$\frac{1}{n}\log |(\phi_n,D_n)|_1\to R_1$ and 
$\frac{1}{n}\log |(\phi_n,D_n)|_2\to R_2$\footnote{The definitions of 
$|(\phi_n,D_n)|_1$ and $|(\phi_n,D_n)|_2$ are given in the end of Section II.}.
A rate pair $(R_1,R_2)$ is weakly deterministically (stochastically) achievable when
there exists a weakly secure sequence of deterministic (stochastic) codes $\{(\phi_n,D_n)\}$
such that 
$\frac{1}{n}\log |(\phi_n,D_n)|_1\to R_1$ and 
$\frac{1}{n}\log |(\phi_n,D_n)|_2\to R_2$.
Then, we denote the set of strongly deterministically (stochastically) achievable rate pair $(R_1,R_2)$ by ${\cal R}_{s,d}$ (${\cal R}_{s,s}$).
In the same way, we denote the set of weakly deterministically (stochastically) achievable rate pair $(R_1,R_2)$ by ${\cal R}_{w,d}$ (${\cal R}_{w,s}$).

\begin{thm}\Label{Converse}
We have the following characterization.
\begin{align}
{\cal R}_{w,d}  \subset \overline{{\cal C}}
,\quad 
{\cal R}_{s,s} \subset \overline{{\cal C}} .\Label{Con2} 
\end{align}
\hfill $\square$\end{thm}

\begin{thm}\Label{TH3}
A rate pair $(R_1,R_2)$ is strongly deterministically achievable
when 
the condition \eqref{OfIm} holds and
there exists a distribution $P \in {\cal P}(\cX)$
such that $\zeta_1(P)>0$ and 
\begin{align}
0 < R_1-R_2< I(X;Y)_P< R_1 < H(X)_P \Label{NHO}.
\end{align}
\hfill $\square$\end{thm}

In fact, 
the condition $R_1-R_2< I(X;Y)_P$ corresponds to Verifiable condition (A),
the condition $I(X;Y)_P< R_1 $ does to Non-decodable condition (B),
and
the conditions $R_1 < H(X)_P $ and $ \zeta_1(P)>0$ do to Non-cheating condition for dishonest Alice (D).
Theorems \ref{Converse} and \ref{TH3} are shown in Sections \ref{S-C} and \ref{S-K}, respectively.
We have the following corollaries from Theorems \ref{Converse} and \ref{TH3}. 

\begin{corollary}\Label{Cor46}
When the condition \eqref{OfIm} holds,
we have the following relation for $K=(s,s),(s,d),(w,d)$.
\begin{align}
\overline{{\cal R}_{K}} \cap \{(R_1,R_2) | R_1 \le H_0\}
=\overline{{\cal C}}\cap \{(R_1,R_2) | R_1 \le H_0\}.
\Label{TDAB}
\end{align}
\hfill $\square$
\end{corollary}

\begin{proof}
Since $\overline{{\cal R}_{s,d}}\subset \overline{{\cal R}_{s,s}},\overline{{\cal R}_{w,d}}$,
Theorem \ref{Converse} guarantees that
\begin{align}
\overline{{\cal R}_{K}} \cap \{(R_1,R_2) | R_1 \le H_0\}
\subset \overline{{\cal C}}\cap \{(R_1,R_2) | R_1 \le H_0\}
\Label{EB1}
\end{align}
for $K=(s,s),(s,d),(w,d)$.

Applying Theorem \ref{TH3} to the case with $P=P_{\max}$ we find that
any inner point of $\overline{{\cal C}}\cap \{(R_1,R_2) | C(W) \le R_1 \le H_0\}$ is achievable.
For any point $(R_1,R_2) $ of
$\overline{{\cal C}}\cap \{(R_1,R_2) | 0 \le R_1 \le C(W) \}$, 
there exist a real number $p \in [0,1]$ and
an inner point $(R_1',R_2') $ of $\overline{{\cal C}}\cap \{(R_1,R_2) | C(W) \le R_1 \le H_0\}$
such that 
$(R_1,R_2)=(p R_1',p R_2') $.
The rate pair $(R_1',R_2') $ is strongly deterministically achievable,
For $n$ transmissions, 
we choose a code for $p n$ transmissions to achieve the rate pair $(R_1',R_2') $ strongly deterministically,
and we choose a code for $(1-p) n$ transmissions with the rate $(0,0)$.
Then, the concatenated code strongly deterministically achieves the rate pair $(R_1,R_2)$.
Hence, we obtain the relation opposite to \eqref{EB1}.
\end{proof}

\begin{corollary}\Label{Cor56}
When the condition \eqref{OfIm} holds, and
the uniform distribution $P_{\uni}$ on ${\cal X}$ satisfies the condition $I(X;Y)_{P_{\uni}}=C(W)$,
we have the following relations.
\begin{align}
\overline{{\cal R}_{s,s}} =
\overline{{\cal R}_{s,d}} =
\overline{{\cal R}_{w,d}} =\overline{{\cal C}}.\Label{TDAC}
\end{align}
In this case, these capacity regions can be characterized by two real numbers $\log |\cX|$ and $C(W)$ as Fig. \ref{FF1}.
\hfill $\square$
\end{corollary}
For example, when the channel $W$ is additive, the assumption of Corollary \ref{Cor56} holds.

\begin{figure}[t]
\begin{center}
  \includegraphics[width=0.7\linewidth]{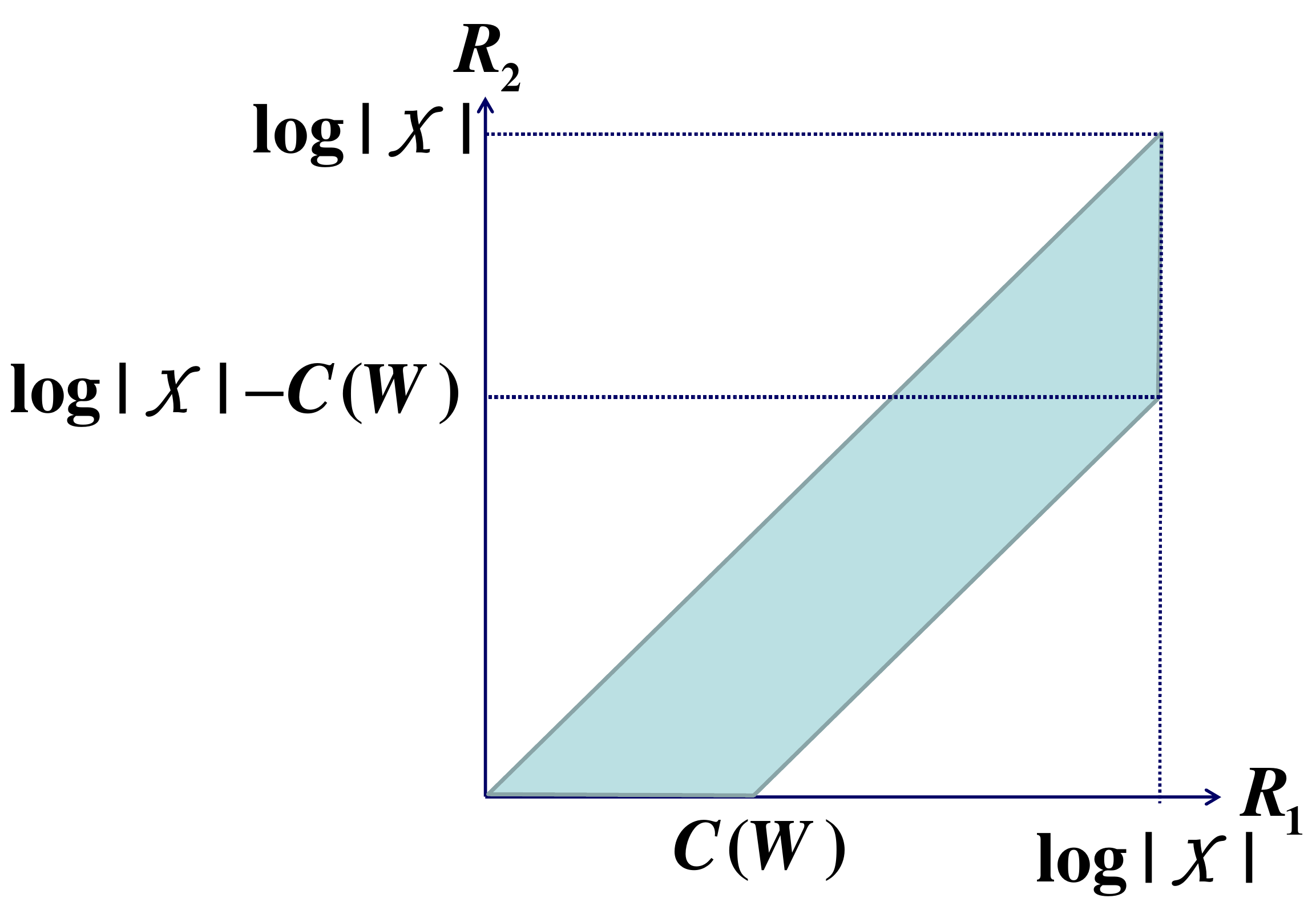}
  \end{center}
\caption{Capacity Region given in Corollary \ref{Cor56}.}
\Label{FF1}
\end{figure}

\begin{proof}
The assumption implies that $H_0= \log |\cX|=H(X)_{P_{\uni}}$.
Due to Theorem \ref{Converse}, for $K=(s,s),(s,d),(w,d)$,
any rate pair $(R_1,R_2)$ in ${\cal R}_{K}$ satisfies the condition $R_1 \le \log |\cX|=H(X)_{P_{\uni}}=H_0$.
Hence, the condition \eqref{TDAB} is rewritten as \eqref{TDAC}.
\end{proof}

\begin{corollary}\Label{Cor66}
When the condition \eqref{OfIm} holds, and
every distribution $P \in {\cal P}_0$ satisfies the condition \eqref{IMC}, i.e., $\zeta_1(P)>0$,
we have the relations \eqref{TDAC}.
In this case, these capacity regions can be characterized as Fig. \ref{FF2}.
\hfill $\square$
\end{corollary}

\begin{proof}
Due to Corollary \ref{Cor46}, it is sufficient to show 
\begin{align}
\overline{{\cal R}_{K}} \cap \{(R_1,R_2) | H_0\le R_1 \le \log |\cX|\}
=\overline{{\cal C}}\cap \{(R_1,R_2) | H_0\le R_1 \le \log |\cX| \}
\Label{EB5}
\end{align}
for $K=(s,s),(s,d),(w,d)$.
Since $\overline{{\cal R}_{s,d}}\subset \overline{{\cal R}_{s,s}},\overline{{\cal R}_{w,d}}$,
Theorem \ref{Converse} guarantees the relation
$\subset$ in \eqref{EB5} for $K=(s,s),(s,d),(w,d)$.
It  is sufficient to show the achievability of any rate pair in 
${{\cal C}}\cap \{(R_1,R_2) | H_0\le R_1 \le \log |\cX| \}$.

The combination of the assumption and Theorem \ref{TH3} guarantees that 
the region $\cup_{P \in {\cal P}_0}
\{(R_1,R_2)| H_0 \le R_1 \le H(X)_P ,~
R_1- I(X;Y)_P\le R_2\le R_1\})$
is strongly deterministically achievable.
Hence, the relation $\supset$ in \eqref{EB5} follows from Lemma \ref{LA2}
if we can show the following statement;
Any convex combination of elements of the region $\cup_{P \in {\cal P}_0}
\{(R_1,R_2)| H_0 < R_1 < H(X)_P ,~
R_1- I(X;Y)_P< R_2< R_1\})$
is strongly deterministically achievable.

To show the above required statement, we assume that two sequences 
$\{(\phi_n,D_n)\}$ and $\{(\phi_n',D_n')\}$ 
of deterministic codes 
are strongly secure.
Then, we define the concatenation $\{(\phi_{2n}'',D_{2n}'')\}$ 
as follows.
When $\phi_n$($\phi_n'$) is given as a map from ${\cal M}$(${\cal M}'$) to ${\cal X}^n$,
the encoder $\phi_{2n}'' $ is given as a map 
from $(m,m')\in {\cal M} \times {\cal M}'$ to $(\phi_n(m),\phi_n'(m'))\in {\cal X}^{2n}$. 
The decoder $D_{2n}''$ is given as a map from $(y_1, \ldots, y_{2n})\in {\cal Y}^{2n}$
to $(D_{n}(y_1, \ldots, y_n),D_{n}'(y_{n+1}, \ldots, y_{2n})) \in {\cal M}^{\sL} \times {{\cal M}'}^{\sL'}$.
We have 
$\epsilon_A(\phi_{2n}'',D_{2n}'') \le \epsilon_A(\phi_n,D_n) + \epsilon_A(\phi_n',D_n') $
because 
the code $(\phi_{2n}'',D_{2n}'')$ is correctly decoded when 
both codes $(\phi_n,D_n)$ and $(\phi_n',D_n')$ are correctly decoded.
Since the message encoded by $\phi_{2n}''$ is correctly decoded only when 
both messages encoded by encoders $\phi_n$ and $\phi_n'$ are correctly decoded,
we have
$\delta_B(\phi_{2n}'') \le \min (\delta_B(\phi_n) ,\delta_B(\phi_n')) $.
Alice can cheat the decoder $D_{2n}''$
only when 
Alice cheats one of the decoders $D_{n}$ and $D_{n}'$.
Hence,
$\delta_D(D_{2n}'') \le \min( \delta_D(D_n),\delta_D(D_n'))$.
Therefore, 
the concatenation $\{(\phi_{2n}'',D_{2n}'')\}$ is also strongly secure.
That is,
any convex combination of elements of the region $\cup_{P \in {\cal P}_0}
\{(R_1,R_2)| H_0 \le R_1 \le H(X)_P ,~
R_1- I(X;Y)_P\le R_2\le R_1\})$
is strongly deterministically achievable.
\end{proof}

\begin{figure}[t]
\begin{center}
  \includegraphics[width=0.7\linewidth]{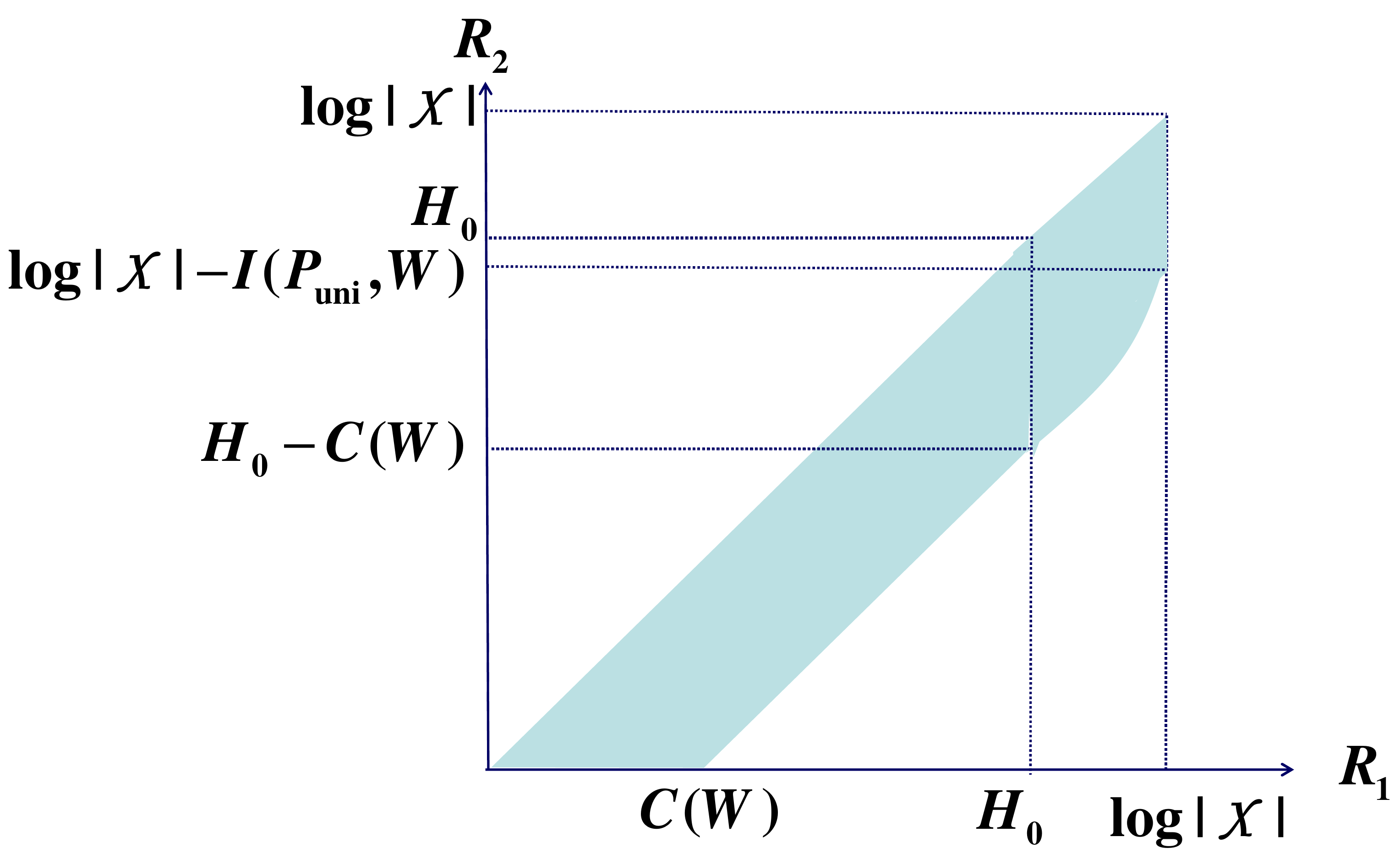}
  \end{center}
\caption{Capacity Region given in Corollary \ref{Cor66}.}
\Label{FF2}
\end{figure}   

\section{Proof of Converse Theorem}\Label{S-C}
In order to show Theorem \ref{Converse}, we prepare the following lemma.
\begin{lem}\Label{L4}
For $X^n=(X_{1}, \ldots, X_{n})$, 
we choose the joint distribution $P_{X^n}$.
Let $Y^n=(Y_{1}, \ldots, Y_{n})$ be the channel output variables of the inputs $X^n$ via the channel $W$.
Then, we have
\begin{align}
I(X^n;Y^n) & \le \sum_{j=1}^n I( X_{j};Y_j ),\Label{LLP3}\\
H(X^n) & \le \sum_{j=1}^n H( X_{j}).\Label{LLP2}
\end{align}
\hfill $\square$
\end{lem}
Although this lemma is not so difficult,
for readers' convenience, we give the proof in Appendix \ref{A1}.

\noindent{\it Proof of Theorem \ref{Converse}:}\quad
The proof of Theorem \ref{Converse} is composed of two parts.
The first part is the evaluation of $R_1$.
The second part is the evaluation of $R_1-R_2$.
The key point of the first part is the use of \eqref{LLP3} in Lemma \ref{L4}.
The key point of the second part is the meta converse for list decoding \cite[Section III-A]{Haya}. 

\noindent{\bf Step 1:} Preparation.

\noindent We show Theorem \ref{Converse} by showing the following relations
\begin{align}
{\cal R}_{w,d}  &\subset 
\cup_{P\in {\cal P}({\cal U}\times {\cal X})} 
\{ (R_1,R_2) | 0 \le R_1-R_2\le I(X;Y|U)_P , ~R_1 \le H(X|U)_P ,~0\le R_1,~0\le R_2\}  
 ,\Label{Con1} \\
{\cal R}_{s,s} &\subset 
\cup_{P\in {\cal P}({\cal U}\times {\cal X})} 
\{ (R_1,R_2) | 0 \le R_1-R_2\le I(X;Y|U)_P , ~R_1 \le H(X|U)_P ,~0\le R_1,~0\le R_2\} .\Label{Con2B} 
\end{align}

Assume that a sequence of deterministic codes $\{(\phi_n,D_n)\}$ is 
weakly secure.
We assume that 
$R_i:=\lim_{n\to \infty}\frac{1}{n}\log |(\phi_n,D_n)|_i$ converges for $i=1,2$.
For the definition of $|(\phi_n,D_n)|_i$,
see the end of Section \ref{S2-1B}.
Letting $M$ be the random variable of the message, 
we define the variables $X^n=(X_1, \ldots, X_n):= \phi_n(M)$. 
The random variables $Y^n=(Y_1, \ldots, Y_n)$ are defined as the output of the channel $W^n$, 
which is the $n$ times use of the channel $W$.
We define the joint distribution $P_n \in {\cal P}({\cal U}\times \cX)$ by $P_n(x,i):= \frac{1}{n} P_{X_i}(x)$ for 
$i=1, \ldots, n$ and $x \in \cX$ by choosing ${\cal U}=\{1, \ldots, n\}$.
Hence, the variable $U$ is subject to the uniform distribution on $\{1, \ldots, n\}$.
Under the distribution $P_{n}$, we denote the channel output by $Y$.
In this proof, 
we use the notations $\sM_n: =|(\phi_n,D_n)|_1$ and $\sL_n:=|(\phi_n,D_n)|_2 $.
Also, instead of $\epsilon_{A}(\phi_n,D_n)$, we employ
$\epsilon_{A}'(\phi_n,D_n):= \sum_{m=1}^{\sM_n} \frac{1}{\sM_n}
\epsilon_{A,m}(\phi_n(m),D_n) $, which goes to zero. 

\noindent{\bf Step 2:} Evaluation of $R_1$.

\noindent 
For a code $(\phi_n,D_n)$, 
we have
\begin{align}
\log |(\phi_n,D_n)|_1
\stackrel{(a)}{\le}&
  H(X^n)
+\epsilon_{A}(\phi_n,D_n)\log |(\phi_n,D_n)|_1+\log 2 \nonumber \\
\stackrel{(b)}{\le}&
n H(X|U)_{P_n}
+\epsilon_{A}(\phi_n,D_n)\log |(\phi_n,D_n)|_1+\log 2 \Label{MGP},
\end{align}
where $(b)$ follows from Lemma \ref{L4}.
Dividing the above by $n$ and taking the limit, we have
\begin{align}
\limsup_{n\to \infty}R_1 - H(X|U)_{P_n} \le 0 \Label{MGP2}.
\end{align}

To show $(a)$ in \eqref{MGP}, we consider the following protocol.
After converting the message $M$ to $X^n$ by the encoder $\phi_n(M)$,
Alice sends the $X^n$ to Bob $K$ times.
Here, we choose $K$ to be an arbitrary large integer.
Applying the decoder $D_n$, Bob obtains $K$ lists that contain up to $\sL^K$ messages.
Among these messages, Bob chooses  
$\hat{M}$ as the element that most frequently appears in the $K$ lists. 
When $ \delta_C(\phi_n,D_n)<1- \epsilon_{A,M}(\phi_n(M),D_n)$ and $K$ is sufficiently large,
Bob can correctly decode $M$ by this method
because $1- \epsilon_{A,M}(\phi_n(M),D_n)$ is the probability that the list contains $M$
and $ \delta_C(\phi_n,D_n)$ is the maximum of the probability that the list
contains $m'\neq M$, i.e., 
the element $M$ has the highest probability to be contained in the list.
Therefore, the failure of decoding is limited to the case when 
$ 1-\delta_C(\phi_n,D_n) \le \epsilon_{A,M}(\phi_n(M),D_n)$,
Since the average of $\epsilon_{A,M}(\phi_n(M),D_n) $ is $\epsilon_{A}'(\phi_n,D_n)$,
Markov inequality guarantees that 
the error probability of this protocol is bounded by 
$\epsilon':=\frac{\epsilon_{A}'(\phi_n,D_n)}{1-\delta_C(\phi_n,D_n)}$.
Fano inequality shows that 
$H(M|\hat{M}) \le \epsilon'\log |(\phi_n,D_n)|_1+\log 2$.
Then, we have
\begin{align}
\log |(\phi_n,D_n)|_1-\epsilon'\log |(\phi_n,D_n)|_1+\log 2
 \le \log |(\phi_n,D_n)|_1- H(M|\hat{M})=I(M;\hat{M}) \le H(X^n),
\end{align}
which implies $(a)$ in \eqref{MGP}.

\noindent{\bf Step 3:} Evaluation of $R_1-R_2$.

\noindent Now, we consider the hypothesis testing with two distributions 
$P(m,y^n):=\frac{1}{\sM_n} W^n(y^n| \phi_n(m))$ 
and $Q(m,y^n):=\frac{1}{\sM_n} \sum_{m=1}^{\sM_n} W^n(y^n| \phi_n(m))$
on ${\cal M}_n\times {\cal Y}^n$, where 
${\cal M}_n:=\{1, \ldots,  \sM_n\}$.
Then, we define the region ${\cal D}_n^*\subset {\cal M}_n\times {\cal Y}^n $
as $\cup_{m_1, \ldots, m_{\sL_n}} \{ m_1, \ldots, m_{\sL_n}\} \times {\cal D}_{m_1, \ldots, m_{\sL_n}}$.
Using the region ${\cal D}_n^*$ as our test, we define 
$\epsilon_Q$ as the error probability to incorrectly support $P$ while the true is $Q$.
Also, we define $\epsilon_P$ as the error probability to incorrectly support $Q$ while the true is $P$.
When we apply the monotonicity for the KL divergence between $P$ and $Q$, 
 dropping the term $\epsilon_{P}\log (1- \epsilon_Q ) $,
we have
\begin{align}
-\log \epsilon_Q 
\le \frac{D(P\|Q)+h(1-\epsilon_{P})}
{1-\epsilon_{P}}.
\Label{MK5}
\end{align}
The meta converse for list decoding \cite[Section III-A]{Haya} 
shows that $\epsilon_Q \le \frac{|(\phi_n,D_n)|_2}{|(\phi_n,D_n)|_1} $
and $\epsilon_P= \epsilon_{A}'(\phi_n,D_n)$.
Since Lemma \ref{L4} guarantees that
$D(P\|Q)= I(X^n;Y^n)\le n I(X;Y|U)_{P_n}$, 
the relation \eqref{MK5} is converted to 
\begin{align}
 \log \frac{|(\phi_n,D_n)|_1}{|(\phi_n,D_n)|_2}
\le \frac{I(X^n;Y^n)+h(1-\epsilon_{A}'(\phi_n,D_n))}
{1-\epsilon_{A}'(\phi_n,D_n)} 
\le \frac{n I(X;Y|U)_{P_n}+h(1-\epsilon_{A}'(\phi_n,D_n))}
{1-\epsilon_{A}'(\phi_n,D_n)} 
\Label{MK1}.
\end{align}
Dividing the above by $n$ and taking the limit, we have
\begin{align}
\limsup_{n\to \infty}R_1-R_2 - I(X;Y|U)_{P_n} \le 0 \Label{MK12}.
\end{align}
Therefore, combining Eqs. \eqref{MGP2} and \eqref{MK12},
we obtain Eq. \eqref{Con1}.

\noindent{\bf Step 4:} Proof of Eq. \eqref{Con2B}.

\noindent 
Assume that a sequence of stochastic codes $\{(\phi_n,D_n)\}$ is 
strongly secure.
Then, there exists a sequence of deterministic encoders $\{\phi_n'\}$
such that
$\epsilon_{A}(\phi_n',D_n) \le \epsilon_{A}(\phi_n,D_n)$
and
$\delta_C(\phi_n',D_n) \le \delta_D(D_n)$.
Since $\epsilon_{A}(\phi_n',D_n)\to 0$ and 
$\delta_C(\phi_n',D_n)\to 0$, 
combining Eq. \eqref{Con1}, we have Eq. \eqref{Con2B}.
\endproof



\section{Proof of direct theorem}\Label{S-K}
Here, we prove the direct theorem (Theorem \ref{TH3}).
The proof is based on the random coding.
In this proof, we prepare three lemmas, 
Lemmas \ref{LL12}, \ref{LL10}, and \ref{LL11}.
Using Lemmas \ref{LL10}, and \ref{LL11}, we 
extract an encoder $\phi_n$ and messages $m$ with small decoding error probability
to satisfy two conditions, which will be stated as 
the conditions \eqref{CC1} and \eqref{CC2}.
Then, using these two conditions, 
we show that the code satisfies Non-cheating condition for dishonest Alice (D)
and Non-decodable condition (B).
In particular, Lemma \ref{LL12} is used
to derive Non-cheating condition for dishonest Alice (D).

\subsection{Preparation}
To show Theorem \ref{TH3}, we prepare notations and basic facts.
Assume that $R_1$ and $R_2$ satisfies the condition \eqref{NHO}.
First, given a real number $R_3 <R_1$, we fix 
the size of message $\sM_n:=2^{n R_1}$,
the list size $\sL_n:=2^{n R_2}$,
and a number $\sM_n':=2^{n R_3}$, which is smaller than 
the message size $\sM_n$.
Then, we prepare the decoder used in this proof as follows.
\begin{define}[Decoder $D_{\phi_n}$]\Label{Def1}
Given a distribution $P$ on ${\cal X}$, 
we define the decoder $D_{\phi_n}$ for a given encoder $\phi_n$ (a map from
$\{1, \ldots, \sM_n\}$ to ${\cal X}^n$) in the following way.
We define the subset ${\cal D}_{x^n}:= \{ y^n| W_{x^n} (y^n) \ge \sM_n' W_{P^n}^{n}(y^n) \}$.
Then, for $y^n \in {\cal Y}^n$, 
we choose up to $\sL_n$ elements $i_1, \ldots, i_{\sL_n'}$ $(\sL_n'\le \sL_n)$
as the decoded messages
such that $ y^n \in {\cal D}_{\phi_n(i_j)}$ for $j=1, \ldots, \sL_n'$.
\hfill $\square$
\end{define}

For $x^n,{x^n}'\in {\cal X}^n$, we define
\begin{align}
F^n(x^n,{x^n}'|P)&:=\sum_{i=1}^n F(x^n_i,{x^n_i}'|P)\Label{I11},
\end{align}
and define 
$d(x^n,{x^n}')$ to be the number of 
$k$ such that $x_k \neq {x_k}'$.
In the proof of Theorem \ref{TH3}, 
we need to extract an encoder $\phi_n$ and elements $m \in {\cal M}_n$ that satisfies the following two conditions for $R_3,\epsilon_1,\epsilon_2 >0$;
\begin{align}
F^n(\phi_n(m),\phi_n(m)|P) &< n (I(P,W)+\epsilon_1) \Label{CC1} \\
F^n(\phi_n(m),\phi_n(j)|P) &< n (R_3-\epsilon_2) 
\hbox{ for }\forall j\neq m. \Label{CC2} 
\end{align}
For this aim, given a code $\phi_n$ and real numbers,
$R_3,\epsilon_1,\epsilon_2 >0$, we define the functions 
$\eta_{\phi_n,\epsilon_1}^A$
and
$\eta_{\phi_n,\epsilon_2,R_3}^C$
from ${\cal M}_n $ to $\{0,1\}$ as
\begin{align}
\eta_{\phi_n,\epsilon_1}^A(m) &:=
\left\{
\begin{array}{ll}
0 & \hbox{ when \eqref{CC1} holds} \\
1 & \hbox{ otherwise }
\end{array}
\right. \Label{De1}\\
\eta_{\phi_n,\epsilon_2,R_3}^C(m) &:=
\left\{
\begin{array}{ll}
0 & \hbox{ when \eqref{CC2} holds} \\
1 & \hbox{ otherwise. }
\end{array}
\right. \Label{De2}
\end{align}

As shown in Section \ref{S7-C}, we have the following lemma.
\begin{lem}\Label{LL12}
For arbitrary real numbers $\epsilon_1>0$ and $R_3<I(P,W)$, 
we choose  $ \epsilon_2:=\epsilon_1+ \frac{\zeta_2(P)}{\zeta_1(P)}  (I(P,W)- R_3+\epsilon_1)$.
When a code $\tilde{\phi}$ defined in the message set $\tilde{{\cal M}}_n$
satisfies
\begin{align}
\max_{m' (\neq m) \in \tilde{{\cal M}}_n} 
F^n(\tilde{\phi}_n(m),\tilde{\phi}_n(m')|P) &< n (R_3-\epsilon_2) \Label{E41}\\
F^n(\tilde{\phi}_n(m),\tilde{\phi}_n(m)|P) &< n (I(P,W)+\epsilon_1)  
\Label{E42}
\end{align}
for an element $ m \in \tilde{{\cal M}}_n$,
we have
\begin{align}
&\delta_{D,m}(D_{\tilde{\phi}_n}) 
\le
 \frac{V(W)}{ n (\epsilon_1
- \frac{\zeta_2(P)\sqrt{2 V}}{\zeta_1(P)\sqrt{ n}} 
) ^2}.
\end{align}
\hfill $\square$
\end{lem}

\subsection{Proof of Theorem \ref{TH3}}
\noindent {\bf Step 1}: Lemmas related to random coding.

\noindent To show Theorem \ref{TH3},
we assume that the variable $\Phi_n(m)$ for $m \in {\cal M}_n$
is subject to the distribution $P^n$ independently.
Then, we have the following two lemmas, which are shown later.
In this proof, we treat the code $\Phi_n$ as a random variable.
Hence, the expectation and the probability for this variable
are denoted by $\rE_{\Phi_n} $ and ${\rm Pr}_{\Phi_n}$, respectively.

\begin{lem}\Label{LL10}
When 
\begin{align}
I(P,W)&>R_3,\quad R_3 \ge R_1-R_2,\Label{C4}
\end{align}
we have the average version of Verifiable condition (A), i.e., 
\begin{align}
\lim_{n \to \infty}
\rE_{\Phi_n} 
\sum_{m=1}^{\sM_n}
\frac{1}{\sM_n} 
\epsilon_{A,m}(\Phi_n,D_{\Phi_n}) 
=0
 \Label{ER1}.
\end{align}
\hfill $\square$
\end{lem}

\begin{lem}\Label{LL11}
When 
\begin{align}
H(P)&>R_1,\Label{CG}
\end{align}
for every real number $\epsilon_2>0$, we have
\begin{align}
\lim_{n \to \infty}
\rE_{\Phi_n} 
\sum_{m=1}^{\sM_n}
\frac{1}{\sM_n} 
\eta_{\Phi_n,\epsilon_2,R_3}^C(m) =0
\Label{ER3}.
\end{align}
\hfill $\square$
\end{lem}

\noindent {\bf Step 2}: 
Extraction of an encoder $\phi_n$ and messages $m$ with small decoding error probability 
that satisfies the conditions \eqref{CC1} and \eqref{CC2}.

\noindent In this proof, we set the parameters 
$\epsilon_0$, $\epsilon_1$, $\epsilon_2$, $\epsilon_3$, and $R_3$ in the following way so that
the conditions in Lemmas \ref{LL12}, \ref{LL10}, and \ref{LL11} are satisfied.
\begin{align}
\epsilon_0&:= H(P)-R_1 \\
\epsilon_1&:=\frac{I(P,W)-R_1+R_2}{2}\\
R_3&:= I(P,W)-\epsilon_1 \\
\epsilon_2&:=\big(1+ \frac{2\zeta_2(P)}{\zeta_1(P)} \big) \epsilon_1
=\epsilon_1+ \frac{\zeta_2(P)}{\zeta_1(P)}  (I(P,W)- R_3+\epsilon_1) \\ 
\epsilon_3 &:=\min(\epsilon_1,\frac{R_1-I(P,W)}{3}).
\end{align}
The law of large number guarantees that
\begin{align}
&\lim_{n \to \infty}
\rE_{\Phi_n} 
\sum_{m=1}^{\sM_n}
\frac{1}{\sM_n} 
\eta_{\Phi_n,\epsilon_3}^A(m) =0.
\Label{ER4}  
\end{align}

Due to Eq. \eqref{ER4}, 
and Lemmas \ref{LL10} and \ref{LL11},
there exist a sequence of codes $\phi_n$ and 
a sequence of real numbers $\epsilon_{4,n}>0$
such that
$\epsilon_{4,n} \to 0$ and 
\begin{align}
\sum_{m=1}^{\sM_n}
\frac{1}{\sM_n} 
\Big(\epsilon_{A,m}(\phi_n,D_{\phi_n}) 
+\eta_{\phi_n,\epsilon_3}^A(m) 
+\eta_{\phi_n,\epsilon_2,R_3}^C(m) 
\Big)
& \le \frac{\epsilon_{4,n}}{3} \Label{BGC}. 
\end{align}
Due to Eq. \eqref{BGC}, Markov inequality guarantees that
there exist $2\sM_n/3$ elements 
$\tilde{{\cal M}}_n:=
\{m_1, \ldots, m_{2 \sM_n/3}\}$
such that every element $m \in \tilde{{\cal M}}_n$ satisfies
\begin{align}
\epsilon_{A,m}(\phi_n,D_{\phi_n}) 
+\eta_{\phi_n,\epsilon_3}^A(m) 
+\eta_{\phi_n,\epsilon_2,R_3}^C(m) 
\le \epsilon_{4,n},
\end{align}
which implies that
\begin{align}
\epsilon_{A,m}(\phi_n,D_{\phi_n}) &\le \epsilon_{4,n}  \\
\eta_{\phi_n,\epsilon_3}^A(m) 
&=\eta_{\phi_n,\epsilon_2,R_3}^C(m) =0
\end{align}
because $\eta_{\phi_n,\epsilon_3}^A$ and $\eta_{\phi_n,\epsilon_2,R_3}^C$ take value 0 or 1.

\noindent {\bf Step 3}: Proof of Non-cheating condition for dishonest Alice (D).

\noindent 
Now, we define a code $\tilde{\phi}_n$
on $\tilde{{\cal M}}_n$
as $\tilde{\phi}_n(m):= {\phi}_n(m)$ for $m \in \tilde{{\cal M}}_n $.
Thus, for $ m \neq m'\in \tilde{{\cal M}}_n$,
we have
\begin{align}
\epsilon_{A,m}(\tilde{\phi}_n,D_{\tilde{\phi}_n}) &\le
\epsilon_{A,m}(\tilde{\phi}_n,D_{\phi_n}) \le \epsilon_{4,n} \Label{EF2}\\
F^n(\tilde{\phi}_n(m),\tilde{\phi}_n(m')|P) &< n (R_3-\epsilon_2) \\
F^n(\tilde{\phi}_n(m),\tilde{\phi}_n(m)|P) &< n (I(P,W)+\epsilon_3) \le n (I(P,W)+\epsilon_1)  .
\Label{MUFV}
\end{align}
Therefore, Lemma \ref{LL12} guarantees 
Non-cheating condition for dishonest Alice (D), i.e., 
\begin{align}
&\delta_{D,m}(D_{\tilde{\phi}_n}) 
\le
 \frac{V(W)}{ n (\epsilon_1
- \frac{\zeta_2(P)\sqrt{2 V}}{\zeta_1(P)\sqrt{ n}} 
) ^2}.\Label{EF1}
\end{align}

\noindent {\bf Step 4}: Proof of Non-decodable condition (B).

\noindent 
Eq. \eqref{MUFV} can be rewritten as
\begin{align}
D( W_{\phi_n(m)}^{n}\| W_{P^n}^{n}  )
< n (I(P,W)+\epsilon_3) .\Label{MUF}
\end{align}
In the code $\tilde{\phi}_n$, Eq. \eqref{MUF} implies that
\begin{align}
& I(X^n;Y^n) 
=
\sum_{m\in \tilde{{\cal M}}_n}
P_{\tilde{{\cal M}}_n}(m)
D( W_{\phi_n(m)}^{n}\| W_{P_{\tilde{{\cal M}}_n}}^{n})\nonumber \\
\le &
\Big(\sum_{m\in \tilde{{\cal M}}_n}
P_{\tilde{{\cal M}}_n}(m)
D( W_{\phi_n(m)}^{n}\| W_{P_{\tilde{{\cal M}}_n}}^{n})\Big)
+
D( W_{P_{\tilde{{\cal M}}_n}}^{n}\| W_{P^n}^{n}  ) \nonumber \\
=&
\sum_{m\in \tilde{{\cal M}}_n} P_{\tilde{{\cal M}}_n}(m)
D( W_{\phi_n(m)}^{n}\| W_{P^n}^{n}  )
\le n  (I(P,W)+\epsilon_3),
\Label{MK6}
\end{align}
where $P_{\tilde{{\cal M}}_n}$ is the uniform distribution on $\tilde{{\cal M}}_n$.
Using the formula given in \cite[Theorem 4]{Verdu}\cite[Lemma 4]{HN},
we have
\begin{align}
\delta_{B}(\tilde{\phi}_n)\le& 
\bigg(\sum_{m \in \tilde{{\cal M}}_n} \frac{2}{\sM_n}
W^n_{\tilde{\phi}_n(m)}
\big(\big\{y^n \big| \log W^n_{\tilde{\phi}_n(m)}(y^n)  -\log W^n_{P^n}(y^n) \ge n (I(P,W)+2\epsilon_3)
\big\}\big) \bigg)\nonumber \\
&+\frac{2 \cdot 2^{n(I(P,W)+2\epsilon_3)}}{\sM_n}.\Label{BU2}
\end{align}
Also, we have
\begin{align}
& W^n_{\tilde{\phi}_n(m)}
\Big(\Big\{y^n\Big| \log W^n_{\tilde{\phi}_n(m)}(y^n)  -\log W^n_{P^n}(y^n) \ge 
n (I(P,W)+2\epsilon_3) \Big\}\Big) \nonumber \\
\stackrel{(a)}{\le} &
W^n_{\tilde{\phi}_n(m)}\Big(\Big\{y^n\Big| 
\big|\log W^n_{\tilde{\phi}_n(m)}(y^n)  -\log W^n_{P^n}(y^n) 
-D( W_{\tilde{\phi}_n(m)}^{n}\| W_{P^n}^{n}  ) \big|
\ge  n \epsilon_3 \Big\} \Big)
\stackrel{(b)}{\le} 
\frac{V(W)}{n \epsilon_1^2},\Label{BU1}
\end{align}
where $(a)$ follows from Eq. \eqref{MUF} and $(b)$ follows from 
Chebychev inequality.
Since 
the relation 
\begin{align*}
&  I(P,W)+2\epsilon_3
\le I(P,W)+\frac{2}{3}(R_1-I(P,W)) \\
= & I(P,W)+(R_1-I(P,W))
-\frac{1}{3}(R_1-I(P,W))
= R_1-\frac{1}{3}(R_1-I(P,W))
\end{align*}
implies 
$\frac{2 \cdot 2^{n(I(P,W)+2\epsilon_3)}}{\sM_n} 
\le \frac{2 \cdot 2^{n(R_1-\frac{1}{3}(R_1-I(P,W)))}}{\sM_n} 
= 2 \cdot 2^{-\frac{n}{3}(R_1-I(P,W)))}
\to 0$, the combination of 
\eqref{BU2} and \eqref{BU1} implies that
$\delta_{B}(\tilde{\phi}_n)\to 0$.
Since RHSs of \eqref{EF2} and \eqref{EF1}
go to zero,
we obtain the desired statement.
\endproof

\subsection{Proof of Lemma \ref{LL12}}\Label{S7-C}
The following is the outline of the proof of Lemma \ref{LL12}.
To show Lemma \ref{LL12}, we need to evaluate 
$\delta_{D,m}(D_{\tilde{\phi}_n}) $.
For this evaluation, we evaluate 
$\delta_{C,m}(x^n, D_{\tilde{\phi}_n}) $
when an element $x^n \in {\cal X}^n$
satisfies the condition
\begin{align}
\epsilon_{A,m}(x^n, D_{\tilde{\phi}_n}) \le \frac{1}{2}.\Label{J11}
\end{align}
As Step 1, we evaluate $W^n_{x^n} ({\cal D}_{{x^n}'})$
by using Chebyshev inequality.
Then. in Step 2, we evaluate $F^n(x^n,\tilde{\phi}_n(m')|P)$.
Finally, Step 3 evaluate 
$\delta_{C,m}(x^n, D_{\tilde{\phi}_n}) $
under the condition \eqref{J11} by using the above evaluations.

\noindent {\bf Step 1}: 
Evaluation of $W^n_{x^n} ({\cal D}_{{x^n}'})$.

\noindent 
As the preparation, we evaluate $W^n_{x^n} ({\cal D}_{{x^n}'})$.
The definitions of $F^n(x^n,{x^n}'|P)$ (Eq. \eqref{I11}) 
and $V(W)$ (Eq. \eqref{OfIm})
imply that 
\begin{align}
\mathbb{E}_{x^n} (n R_3 -(\log W_{x^n} (Y^n) -\log W_{P^n}^{n}(Y^n) )) &=
n R_3-F^n(x^n,{x^n}'|P), \\
\mathbb{V}_{x^n} (n R_3 -(\log W_{x^n} (Y^n) -\log W_{P^n}^{n}(Y^n) )) &
\le n V(W).
\end{align}
Hence, applying Chebyshev inequality to 
the variable $n R_3 -(\log W_{x^n} (Y^n) -\log W_{P^n}^{n}(Y^n) )$, we have
\begin{align}
W^n_{x^n} ({\cal D}_{{x^n}'})
&=
W^n_{x^n} (\{ y^n| W_{x^n} (y^n) \ge \sM_n' W_{P^n}^{n}(y^n) \})
\nonumber \\
&=
W^n_{x^n} (\{ y^n|  n R_3 -(\log W_{x^n} (y^n) -\log W_{P^n}^{n}(y^n) )  \le 0  \})
\nonumber \\
&\le 
\frac{n V(W)}{[  n R_3-F^n(x^n,{x^n}'|P) ]_+^2  }\Label{Che},
\end{align}
where $[x]_+:=\max(x,0)$.

\noindent {\bf Step 2}: Evaluation of $F^n(x^n,\tilde{\phi}_n(m')|P)$.

\noindent 
In this step, we evaluate $F^n(x^n,\tilde{\phi}_n(m')|P)$.
In the following, we assume \eqref{J11}. 
Then,
\begin{align}
\frac{1}{2}\le (1-\epsilon_{A,m}(x^n, D_{\tilde{\phi}_n}) )
\stackrel{(a)}{\le} W^n_{x^n} ({\cal D}_{\tilde{\phi}_n(m)}) 
\stackrel{(b)}{\le} \frac{nV(W)}{ [n R_3-F^n(x^n,\tilde{\phi}_n(m)|P)]_+^2},
\end{align}
where $(a)$ follows from the definition of $D_{\tilde{\phi}_n}$ 
(Definition \ref{Def1})
and $(b)$ follows from Eq. \eqref{Che}.

Hence, we have 
\begin{align}
n R_3-F^n(x^n,\tilde{\phi}_n(m)|P)
\le \sqrt{2n V(W)}.\Label{MKO}
\end{align}
Thus,
\begin{align}
&n R_3-\sqrt{2n V(W)}
\stackrel{(a)}{\le}
F^n(x^n,\tilde{\phi}_n(m)|P) \nonumber \\
\stackrel{(b)}{\le} &
F^n(\tilde{\phi}_n(m),\tilde{\phi}_n(m)|P) 
- d(x^n,\tilde{\phi}_n(m))\zeta_1(P) \nonumber \\
\stackrel{(c)}{\le} &
n (I(P,W)+\epsilon_1)
- d(x^n,\tilde{\phi}_n(m))\zeta_1(P),
\end{align}
where $(a)$, $(b)$, and $(c)$ follow from Eq. \eqref{MKO}, 
the combination of the definitions of 
$\zeta_1(P)$ and $d(x^n,\tilde{\phi}_n(m))$,
and Eq. \eqref{E42},
respectively.

Hence,
\begin{align}
d(x^n,\tilde{\phi}_n(m)) \le
\frac{n (I(P,W)- R_3+\epsilon_1)+\sqrt{2n V(W)}}{\zeta_1(P)}.
\Label{E43}
\end{align}
Thus, for $m' \in \tilde{{\cal M}}_n$, we have
\begin{align}
& n R_3-F^n(x^n,\tilde{\phi}_n(m')|P)
\stackrel{(a)}{\ge}
n R_3-F^n(\tilde{\phi}_n(m),\tilde{\phi}_n(m')|P)
-\zeta_2(P) d(x^n,\tilde{\phi}_n(m)) \nonumber \\
\stackrel{(b)}{\ge}
 & n \epsilon_2
-\zeta_2(P) d(x^n,\tilde{\phi}_n(m))\nonumber  \\
\stackrel{(c)}{\ge}
 & n \epsilon_2
- \frac{\zeta_2(P)}{\zeta_1(P)}  (n (I(P,W)- R_3+\epsilon_1)+\sqrt{2n V(W)})
\nonumber \\
\stackrel{(d)}{=}
 & n \epsilon_1
- \frac{\zeta_2(P)}{\zeta_1(P)} \sqrt{2n V(W)},\Label{EE73}
\end{align}
where $(a)$, $(b)$, $(c)$ and $(d)$ follow from 
the combination of the definitions of 
$\zeta_1(P)$ and $d(x^n,\tilde{\phi}_n(m))$,
Eq. \eqref{E41}, 
Eq. \eqref{E43},
the choice of $\epsilon_2$
respectively.
Eq. \eqref{EE73} is the required evaluation of
$F^n(x^n,\tilde{\phi}_n(m')|P)$.

\noindent {\bf Step 3}: Evaluation of $\delta_{C,m}(x^n,D_{\tilde{\phi}_n}) $.

\noindent 
Finally, combining Eqs. \eqref{Che} and \eqref{EE73}, we evaluate $\delta_{C,m}(x^n,D_{\tilde{\phi}_n}) $.
We have
\begin{align}
&\delta_{C,m}(x^n,D_{\tilde{\phi}_n}) 
\stackrel{(a)}{=}
 \max_{m'\neq m}W^n_{x^n} 
({\cal D}_{\tilde{\phi}_n(m')}) \nonumber \\
\stackrel{(b)}{\le}
 & \max_{m'\neq m} \frac{nV(W)}{ 
(n R_3-F^n(x^n,\tilde{\phi}_n(m')|P))^2}
\stackrel{(c)}{\le}
 \frac{nV(W)}{ (n \epsilon_1
- \frac{\zeta_2(P)}{\zeta_1(P)} \sqrt{2n V(W)})^2}
= \frac{V(W)}{ n (\epsilon_1
- \frac{\zeta_2(P)\sqrt{2 V(W)}}{\zeta_1(P)\sqrt{ n}} 
) ^2},
\end{align}
where $(a)$, $(b)$, and $(c)$ follow from 
the definition of $\delta_{C,m}$ (Eq. \eqref{EE8}),
Eq. \eqref{Che}, and
Eq. \eqref{EE73}, 
respectively.
Taking the maximum of $\delta_{C,m}(x^n,D_{\tilde{\phi}_n})$ under the condition \eqref{J11},
we obtain the required inequality.

\endproof

\subsection{Proof of Lemma \ref{LL10}}\Label{S7-D}
We show Lemma \ref{LL10} by employing an idea similar to \cite{Verdu,HN}. 
First, we show the following lemma.
\begin{lem}\Label{NMU}
We have the following inequality;
\begin{align}
\epsilon_A(\Phi_n,D_{\Phi_n})
\le 
\sum_{i=1}^{\sM_n} \frac{1}{\sM_n}
\Big(W_{\Phi_n(i)}({\cal D}_{\Phi_n(i)}^c)
+
\sum_{j\neq i } \frac{1}{\sL_n}
W_{\Phi_n(i)}({\cal D}_{\Phi_n(j)}) \Big)\Label{NML}.
\end{align}
\hfill $\square$
\end{lem}
\begin{proof}
When $i$ is sent, 
there are two cases for incorrect decoding.
The first case is the case that the received element $y$ does not belong to ${\cal D}_{\Phi_n(i)}$.
The second case is the case that there are more than $\sL_n$ elements $i'$ 
to satisfy $y \in {\cal D}_{\Phi_n(i')}$.
The error probability of the first case is given in the first term of Eq. \eqref{NML}.
The error probability of the second case is given in the second term of Eq. \eqref{NML}.
\end{proof}

Taking the average in \eqref{NML} of Lemma \ref{NMU}
with respect to the variable $\Phi_n$, we obtain the following lemma.
\begin{lem}\Label{FHU}
We have the following inequality;
\begin{align}
\rE_{\Phi_n}\epsilon_A(\Phi_n,D_{\Phi_n})
\le 
\sum_{x^n\in {\cal X}^n}P^n(x^n) 
\Big(W_{x^n}^n({\cal D}_{x_n}^c)
+
\frac{\sM_n-1}{\sL_n}
W_{P^n}^n({\cal D}_{x^n})\Big).
\Label{NML2}
\end{align}
\hfill $\square$
\end{lem}
Applying Lemma \ref{FHU}, 
we have
\begin{align}
& \rE_{\Phi_n} \epsilon_A(\Phi_n,D_{\Phi_n})\nonumber  \\
\le &
\rE_{X^n}W^n_{X^n} \big(\big\{y^n\big| 
2^{-n R_3} W^n_{X^n}(y^n)   < W^n_{P^n}(y^n) \big\} \big)\nonumber \\
&+
\rE_{X^n}2^{n (R_1-R_2)}W^n_{P^n}
\big(\big\{y^n\big|
2^{-n R_3} W^n_{X^n}(y^n)  \ge W^n_{P^n}(y^n)
\big\}\big)
\nonumber \\
\stackrel{(a)}{\le}&
\rE_{X^n}W^n_{X^n} \big(\big\{y^n\big| 
\log W^n_{X^n}(y^n)  -\log W^n_{P^n}(y^n) < n R_3\big\} \big)\nonumber \\
&+
\rE_{X^n}2^{n (R_1-R_2)}2^{-n R_3}W^n_{X^n}
\big(\big\{y^n\big| 
2^{-n R_3} W^n_{X^n}(y^n)  \ge W^n_{P^n}(y^n)\big\}\big)
\nonumber \\
\le &
\rE_{X^n}W^n_{X^n} \bigg(\bigg\{y^n\bigg| 
\frac{1}{n} (\log W^n_{X^n}(y^n)  -\log W^n_{P^n}(y^n)) <  R_3\bigg\} \bigg)
+2^{n (R_1-R_2-R_3)}
 \Label{ER1B},
\end{align}
where $(a)$ follows from 
the relation 
$$W^n_{P^n}
\big(\big\{y^n\big|
2^{-n R_3} W^n_{X^n}(y^n)  \ge W^n_{P^n}(y^n)
\big\}\big)
\le
2^{-n R_3}W^n_{X^n}
\big(\big\{y^n\big| 
2^{-n R_3} W^n_{X^n}(y^n)  \ge W^n_{P^n}(y^n)\big\}\big).$$
The variable $\frac{1}{n} (\log W^n_{X^n}(y^n)  -\log W^n_{P^n}(y^n))$ 
is the mean of $n$ independent variables 
that are identical to the variable $\log W_X(Y)-\log W_P(Y)$ whose
average is $I(P,W)> R_3$.
Thus, the law of large number guarantees that the first term in 
\eqref{ER1B} approaches to zero as $n$ goes to infinity.
The second term in 
\eqref{ER1B} also approaches to zero due to 
the assumption \eqref{C4}.
Therefore, we obtain Eq. \eqref{ER1}. 
\endproof

\subsection{Proof of Lemma \ref{LL11}}\Label{S7-E}
The outline of the proof of Lemma \ref{LL11} is the following.
To evaluate the value
$\rE_{\Phi_n} 
\sum_{m=1}^{\sM_n}
\frac{1}{\sM_n} 
\eta_{\Phi_n,\epsilon_2,R_3}^C(m)$, we convert it to the sum of certain probabilities.
We evaluate these probabilities by excluding a certain exceptional case.
That is, we show that the probability of the exceptional case is small and 
these probabilities under the condition to exclude the exceptional case is also small.
The latter will be shown by evaluating a certain conditional probability.
For this aim, we introduce a new function $G(s|P)$ and prepare a new lemma.

\noindent {\bf Step 1}: Preparation.

\noindent 
For $s >0$, we define
\begin{align}
G(s,x|P)&:=
\log \sum_{x' \in {\cal X}}P(x')2^{s F(x,x'|P)} \\
G(s|P)&:=
\sum_{x \in {\cal X}}P(x)
G(s,x|P) \\
G^n(s,x^n|P)& :=\sum_{i=1}^n G(s,x^n_i|P).
\end{align}
Since the function $s \mapsto G(s,x|P)$ is strictly convex, the function $s \mapsto G(s|P)$ is strictly convex.
Hence, we have the following lemma.

\begin{lem}\Label{L55}
When $W_x\neq W_{x'}$ for $x\neq x'$, 
The value $s I(P,W) -G(s|P)-H(P)$ is negative and continuous for $s>0$. 
It converges to zero as $s$ goes to infinity.
Also,
$\sup_{s> 0}-R_1+sR -G(s|P) >0$
for $ R_1< H(P)$.
\hfill $\square$\end{lem}
\begin{proof}
Since $ F(x,x|P)  >F(x,x'|P)$  for $x\neq x'$, we have
\begin{align}
G(s|P)= \sum_{x\in {\cal X}}P(x) 
\Big(\log P(x) +s F(x,x|P) +
\log \Big(1+ \sum_{x'\neq x} \frac{P(x')}{P(x)}e^{-s (F(x,x|P) -F(x,x'|P))}
\Big)\Big).
\end{align}
Hence, the relation $\sum_{x\in {\cal X}}P(x) F(x,x|P)= I(P,W)$ implies that
\begin{align}
s I(P,W) -G(s|P)-H(P)
=-\sum_{x\in {\cal X}}P(x)\log 
\Big(1+ \sum_{x'(\neq x)\in {\cal X}} \frac{P(x')}{P(x)}e^{-s (F(x,x|P) -F(x,x'|P))}
\Big)< 0.
\end{align}
When $s \to \infty$, the above value goes to zero.

Since $ -R_1> -H(P)$, we have
$\lim_{s\to 0}-R_1+sR  -G(s|P)
>\lim_{s\to 0}-H(P)+s I(P,W) -G(s|P)=0$.
Hence, due to the continuity for $s$,
there exists $s>0$ such that $-R_1+sR -G(s|P) >0$.
Then, the proof is completed.
\end{proof}
\noindent {\bf Step 2}: Evaluation of a certain conditional probability

\noindent 
Due to Lemma \ref{L55} and Eq. \eqref{CG},
we can choose $s>0$ such that $-R_1+s(R_3-\epsilon_2) >G(s|P)$.
We set $\epsilon_5 := \frac{1}{2}(-R_1+s(R_3-\epsilon_2)-G(s|P) )>0$.
We define two conditions
$A_{n,i}$ and $B_{n,i}$ for the encoder $\Phi_n$ as
\begin{description}
\item[$A_{n,i}$]
$G^n(s, \Phi_n(i)|P) < n (- R_1+s(R_3-\epsilon_2)-\epsilon_5)$.
\item[$B_{n,i}$]
$\exists j\neq i, F^n(\Phi_n(i),\Phi_n(j)|P) \ge n (R_3-\epsilon_2)$.
\end{description}
The aim of this step is the evaluation of the 
conditional probability $\Pr_{\Phi_n} (B_{n,i}|A_{n,i})$ that expresses
the probability that 
the condition $B_{n,i}$ holds under the condition $A_{n,i}$.

We choose $j \neq i$.
When the fixed variable $\Phi_n(i)$ satisfies the condition $A_{n,i}$,
Markov inequality implies that
$$
\Pr_{\Phi_n(j)|\Phi_n(i)} \Big( F^n(\Phi_n(i),\Phi_n(j)|P) \ge n (R_3-\epsilon_2)
\Big) \le 2^{G^n(s, \Phi_n(i)|P) - sn (R_3-\epsilon_2)},
$$
where $\Pr_{\Phi_n(j)|\Phi_n(i)}$ is the probability for the random variable $\Phi_n(j)$ with
the fixed variable $\Phi_n(i)$.
Hence, when the fixed variable $\Phi_n(i)$ satisfies the condition $A_{n,i}$,
\begin{align}
\Pr_{\Phi_{n,i,c} |\Phi_n(i)} (B_{n,i})
\le &
\sum_{j (\neq i)\in {\cal M}_n}
\Pr_{\Phi_n(j)|\Phi_n(i)} \Big( F^n(\Phi_n(i),\Phi_n(j)|P) \ge n (R_3-\epsilon_2)
\Big) \nonumber \\
\le &
(2^{n R_1}-1) 2^{G^n(s, \Phi_n(i)|P) - sn (R_3-\epsilon_2)}
\le
2^{G^n(s, \Phi_n(i)|P) - sn (R_3-\epsilon_2)+nR_1}
\le 2^{-n\epsilon_5},
\end{align}
where $ \Phi_{n,i,c}$ expresses the random variables $\{\Phi_n(j)\}_{j \neq i}$.
Then, we have
\begin{align}
\Pr_{\Phi_n} (B_{n,i}|A_{n,i}) \le 2^{-n\epsilon_5}.
\Label{N2}
\end{align}

\noindent {\bf Step 3}: Evaluation of 
$\rE_{\Phi_n} 
\sum_{m=1}^{\sM_n}
\frac{1}{\sM_n} 
\eta_{\Phi_n,\epsilon_2,R_3}^C(m)$.

\noindent 
The quantity $\rE_{\Phi_n} 
\sum_{m=1}^{\sM_n}
\frac{1}{\sM_n} 
\eta_{\Phi_n,\epsilon_2,R_3}^C(m)$ can be evaluated as
\begin{align}
&\rE_{\Phi_n} 
\sum_{m=1}^{\sM_n}
\frac{1}{\sM_n} 
\eta_{\Phi_n,\epsilon_2,R_3}^C(m) \nonumber \\
=& 
\frac{1}{\sM_n} 
\rE_{\Phi_n} 
|\{ i | B_{n,i} \hbox{ holds. }\}|
=
\sum_{i=1}^{\sM_n}
\frac{1}{\sM_n} \Pr_{\Phi_n}
(B_{n,i}
)\nonumber  \\
\le &
\sum_{i=1}^{\sM_n}
\frac{1}{\sM_n} 
(\Pr_{\Phi_n} (A_{n,i}) \Pr_{\Phi_n} (B_{n,i}|A_{n,i}) 
+ (1-\Pr_{\Phi_n} (A_{n,i}))
)\nonumber  \\
\stackrel{(a)}{\le}
 &
2^{-n\epsilon_5}
+
\sum_{i=1}^{\sM_n}
\frac{1}{\sM_n} 
(1-\Pr (A_{n,i})),\Label{N5}
\end{align}
where $(a)$ follows from Eq. \eqref{N2}.

The random variable $G^n(s, \Phi_n(i)|P)$ can be regarded as 
the $n$-fold i.i.d. extension of
the variable $G(s,X|P)$ whose expectation is $G(s,P)$.
Since the choice of $\epsilon_5$ guarantees that
\begin{align}
G(s,P) <
- R_1+s(R_3-\epsilon_2)-\epsilon_5,
\end{align}
we have
\begin{align}
1-\Pr_{\Phi_n} (A_{n,i})
=\Pr_{\Phi_n} \Big( G^n(s, \Phi_n(i)|P) \ge n (- R_1+s(R_3-\epsilon_2)-\epsilon_5) 
\Big) 
\to 0.
\Label{N3}
\end{align}
Hence, the combination of Eqs. \eqref{N5} and \eqref{N3}
implies the desired statement.
\endproof

\section{Conclusion}
We have  proposed a new concept, secure list decoding, 
which imposes additional requirements on top of those
of conventional list decoding.
This scheme has three requirements.
Verifiable condition (A),
Non-decodable condition (B), and Non-cheating condition.
Verifiable condition (A) means that the message sent by Alice (sender) is contained in the list output by Bob (receiver).
Non-decodable condition (B) means that
Bob cannot uniquely decode Alice's message.
Non-cheating condition has two versions. 
One is the condition (C) for honest Alice.
The other is the condition (D) for dishonest Alice.
Since there is a possibility that Alice uses a different code, 
we need to guarantee the impossibility of cheating 
even for such a dishonest Alice.
In this paper, we have shown the existence of a code to satisfy these three conditions.
Also, we have defined the capacity region as the possible rate pair of the rates of the message and the list,
and have derived the capacity region under a proper condition.
Then, we have clarified the capacity region under several conditions.
Fortunately, since additive noise channels satisfy these conditions, the capacity region is determined in this case.
In general channels, we have clarified only a part of the capacity region.
Therefore, the perfect characterization of the capacity region is an interesting open problem.

As another contribution, we have constructed a protocol for bit commitment
from the secure list decoding. 
However, this conversion protocol is not efficient. That is,
in this protocol, the conversion rate from the size of the secure list decoding to 
the size of  bit commitment is quite small.
Hence, it is a remaining problem to construct a more efficient conversion protocol.
Further, it is not clear whether there is a protocol to convert bit commitment to secure list decoding.
The existence of such a protocol is another interesting open problem.

Since the constructed code in this paper is not practical,
it is needed to construct practical codes for secure list decoding.
Fortunately, the existing study \cite{Guruswami} systematically constructed several types of 
codes for list decoding with their algorithms.
While their code construction is practical,
in order to use their constructed code for secure list decoding,
we need to clarify their security parameters, i.e., 
the non-decodable parameter $\delta_B$ and the non-cheating parameter $\delta_D$
in addition to the decoding error probability $\epsilon_A$.
It is a practical open problem to calculate these security parameters of their codes.

This paper evaluates the non-decodable parameter $\delta_B$.
Clearly, the mutual information between the message $M$ and Bob's received information $Y$ is not zero
because Bob knows that the message is one element in his obtained list $\{M_1,\ldots, M_{\sL}\}$.
Therefore, it is natural to evaluate the amount of leaked information for the message $M$ to Bob,  which can be evaluated by the mutual information.
Alternatively, as the amount of uncertainty for the message $M$ in Bob's side,
we can discuss the conditional entropy for the message $M$ 
conditioned with Bob's received information $Y$.
Since this kind of study will clarify how much secrecy can be kept in this protocol,
it is another challenging future study.


\section*{Acknowledgments}
The author is grateful to 
Dr. Vincent Tan, Dr. Anurag Anshu, and Dr. Naqueeb Warsi  
for helpful discussions.
In particular, Dr. Vincent Tan suggested Lemma \ref{LA1}
and Dr. Anurag Anshu did the relation with bit commitment.
The work reported here was supported in part by 
Fund for the Promotion of Joint International Research
(Fostering Joint International Research) Grant No. 15KK0007,
the JSPS Grant-in-Aid for Scientific Research 
(A) No.17H01280, (B) No. 16KT0017, 
and Kayamori Foundation of Informational Science Advancement.

\appendices

\section{Proof of Lemma \ref{LA1}}\Label{A01}
The region in the first line is rewritten as
\begin{align}
&\cup_{P\in {\cal P}({\cal U}\times {\cal X})} \{ (R_1,R_2) | 0 < R_1< H(X|U)_P,~
[R_1 - I(X;Y|U)_P]_+ <R_2<R_1\} \nonumber \\
=&\{ (R_1,R_2) | 0 < R_1< \log |{\cal X}|, ~(\inf_{P\in {\cal P}({\cal U}\times {\cal X})} \{ [R_1 - I(X;Y|U)_P]_+| R_1< H(X|U)_P\}) < R_2 < R_1 \} \nonumber  \\
=&\{ (R_1,R_2) | 0 < R_1< \log |{\cal X}|, ~(\min_{P\in {\cal P}({\cal U}\times {\cal X})} \{ [R_1 - I(X;Y|U)_P]_+| R_1\le H(X|U)_P\}) < R_2 < R_1 \} \nonumber  \\
=&\{ (R_1,R_2) | 0 < R_1< \log |{\cal X}|, ~[R_1 - \max_{P\in {\cal P}({\cal U}\times {\cal X})} \{I(X;Y|U)_P| R_1\le H(X|U)_P\}]_+ < R_2 < R_1 \}  .\Label{J1}
\end{align}
When $R_1 \le H_0$, we have
\begin{align}
 [R_1 - \max_{P\in {\cal P}({\cal U}\times {\cal X})} \{I(X;Y|U)_P| R_1\le H(X|U)_P\}]_+
=R_1 -C(W)= \kappa(R_1).
\Label{J2}
\end{align}
Since the set $\{(H(X|U)_P, I(X;Y|U)_P)\}_{P\in {\cal P}({\cal U}\times {\cal X})}$
is convex, 
the function $R_1 \mapsto \max_{P\in {\cal P}({\cal U}\times {\cal X})} \{I(X;Y|U)_P|$\par\noindent$ R_1\le H(X|U)_P\}$
is monotonically decreasing for $R_1 > H_0$.
Hence, when $R_1 > H_0$, 
we have
\begin{align}
& [R_1 - \max_{P\in {\cal P}({\cal U}\times {\cal X})} \{I(X;Y|U)_P| R_1\le H(X|U)_P\}]_+\nonumber \\
=& [R_1 - \max_{P\in {\cal P}({\cal U}\times {\cal X})} \{I(X;Y|U)_P| R_1= H(X|U)_P\}]_+\nonumber \\
=&\min_{P\in {\cal P}({\cal U}\times {\cal X})} \{ [R_1 - I(X;Y|U)_P]_+| R_1= H(X|U)_P\} \nonumber \\
=&\min_{P\in {\cal P}({\cal U}\times {\cal X})} \{ [H(X|U)_P - I(X;Y|U)_P]_+| R_1= H(X|U)_P\} \nonumber \\
=&\min_{P\in {\cal P}({\cal U}\times {\cal X})} \{ H(X|YU)_P | R_1= H(X|U)_P\} 
= \kappa(R_1).
\Label{J3}
\end{align}
Therefore, combining \eqref{J1}, \eqref{J2}, and \eqref{J3},
we find that the region in the first line equals the region in the second line.
\hspace*{\fill}~\QED\par\endtrivlist\unskip

\section{Proof of Lemma \ref{LA2}}\Label{A02}
We state the outline of this proof.
First, we show the following statement;
\begin{description}
\item[(S1)]
Given $R_1 \in [H_0,\log |{\cal X}|]$, 
there exists another distribution $P'\in {\cal P}({\cal U}\times \cX)$ to satisfy the following conditions.
(i) For any $u \in {\cal U}$ with $P_U'(u)>0$, we have $H(X)_{P_{X|U=u}'}\ge H_0$.
(ii) $\kappa(R_1) =H(X|YU)_{P'}$.
(iii) $R_1 = H(X|U)_{P'}$.
\end{description}
Next, we show the following statement for the distribution $P'\in {\cal P}({\cal U}\times \cX)$ given in (S1);
\begin{description}
\item[(S2)]
For any $u \in {\cal U}$ with $P_U'(u)>0$, we have 
$H(X|Y)_{P_{X|U=u}'} =\kappa( H(X)_{P_{X|U=u}'} )$.
\end{description}
Finally, we show the following statement;
\begin{description}
\item[(S3)]
For any $R_1 \in [H_0,\log |{\cal X}|]$,
the point $(R_1,\kappa(R_1))$ can be written as 
a convex combination of elements of 
$\{(H(X)_P,H(X|Y)_P)\}_{P \in {\cal P}_0}$,
\end{description}
which derives the required statement.

\noindent {\bf Step 1}: Proof of (S1).

\noindent 
To show (S1), it is sufficient to show the following statement;
Given $R_1 \in [H_0,\log |{\cal X}|]$, 
we choose a distribution $P\in {\cal P}({\cal U}\times \cX)$ such that
$R_1=H(X|U)_P$ and $\kappa(R_1)= H(X|YU)_P$.
Then, there exists another distribution $P'\in {\cal P}({\cal U}\times \cX)$ to satisfy the following conditions.
(i) For any $u \in {\cal U}$ with $P_U'(u)>0$, we have $H(X)_{P_{X|U=u}'}\ge H_0$.
(ii) $H(X|YU)_P \ge H(X|YU)_{P'}$.
(iii) $H(X|U)_P = H(X|U)_{P'}$.
Notice that the combination of (ii) and $\kappa(R_1)= H(X|YU)_P$ implies that 
$\kappa(R_1)= H(X|YU)_{P'}$ due to the definition of the function $\kappa$.

In the following, we show the existence of $P'$ to satisfy the above condition.
Due to the choice of $P_{XU}$, we have
\begin{align}
R_1=\sum_{u\in {\cal U}}P_U (u)H(X)_{P_{X|U=u}}, \quad
\kappa(R_1)= \sum_{u\in {\cal U}}P_U (u)H(X|Y)_{P_{X|U=u}},
\Label{KL1}
\end{align}
where $P_U$ is the marginal distribution of $P$ with respect to $U$, and
$P_{X|U=u}$ is the conditional distribution on $X$ with condition $U=u$.
Also, $P_{X|U\neq u}$ expresses the conditional distribution on $X$ with condition $U\neq u$.

Assume that there exists an element $u_0\in {\cal U}$ such that
$\lambda:= P_U(u_0)>0$ and $H(X)_{P_{X|U=u_0}}< H_0$.
Hence, we have 
\begin{align}
H(X|U)_P&=\lambda H(X)_{P_{X|U=u_0}}+(1-\lambda)H(X|U)_{P_{X|U\neq u_0}},\Label{KG1}\\
H(X|YU)_P&=\lambda H(X|Y)_{P_{X|U=u_0}}+(1-\lambda)H(X|YU)_{P_{X|U\neq u_0}},\Label{KG2} \\
H(X)_{P_{X|U=u_0}}&< H_0=H(X)_{P_{\max}} \le H(X|U)_{P_{X|U\neq u_0}}.\Label{KG2B} 
\end{align}
We choose $\lambda'\in [0,1]$ such that
\begin{align}
H(X|U)_P&=\lambda' H(X)_{P_{\max}}+(1-\lambda')H(X|U)_{P_{X|U\neq u_0}}.
\Label{KG3}
\end{align}

Since the relation $C(W)= I(X;Y)_{P_{\max}}$ yields 
\begin{align}
H(X)_{P_{\max}}-H(X|Y)_{P_{\max}} & \ge H(X)_{P_{X|U=u_0}}-H(X|Y)_{P_{X|U=u_0}}\nonumber \\
H(X)_{P_{\max}}-H(X|Y)_{P_{\max}} & \ge H(X|U)_{P_{X|U\neq u_0}}-H(X|YU)_{P_{X|U\neq u_0}},
\end{align}
we have the relations
\begin{align}
H(X|Y)_{P_{\max}}- H(X|Y)_{P_{X|U=u_0}} & \le H(X)_{P_{\max}}- H(X)_{P_{X|U=u_0}} \nonumber \\
H(X|YU)_{P_{X|U\neq u_0}}-H(X|Y)_{P_{\max}} & \ge H(X|U)_{P_{X|U\neq u_0}}-H(X)_{P_{\max}},
\end{align}
which imply that
\begin{align}
\frac{H(X|Y)_{P_{\max}}- H(X|Y)_{P_{X|U=u_0}}}{H(X)_{P_{\max}}- H(X)_{P_{X|U=u_0}}}  
\le 1\le 
\frac{H(X|YU)_{P_{X|U\neq u_0}}-
H(X|Y)_{P_{\max}}}{H(X|U)_{P_{X|U\neq u_0}}-H(X)_{P_{\max}}}.
\Label{KG4B}
\end{align}
The combination of \eqref{KG2B} and \eqref{KG4B} implies that
\begin{align}
\frac{H(X|YU)_{P_{X|U\neq u_0}}- H(X|Y)_{P_{X|U=u_0}}}{H(X|U)_{P_{X|U\neq u_0}}- H(X)_{P_{X|U=u_0}}}  \le 
\frac{H(X|YU)_{P_{X|U\neq u_0}}-H(X|Y)_{P_{\max}}}{H(X|U)_{P_{X|U\neq u_0}}-H(X)_{P_{\max}}}.
\Label{KG4}
\end{align}
Therefore, we obtain
\begin{align}
H(X|YU)_P
&\stackrel{(a)}{=}
\lambda H(X|Y)_{P_{X|U=u_0}}+(1-\lambda)H(X|YU)_{P_{X|U\neq u_0}} \nonumber \\
&\stackrel{(b)}{=}
H(X|YU)_{P_{X|U\neq u_0}}-
\frac{H(X|YU)_{P_{X|U\neq u_0}}- H(X|Y)_{P_{X|U=u_0}}}{H(X|U)_{P_{X|U\neq u_0}}- H(X)_{P_{X|U=u_0}}}  
(H(X|U)_{P_{X|U\neq u_0}}-H(X|U)_P)
\nonumber \\
&\stackrel{(c)}{\ge}
H(X|YU)_{P_{X|U\neq u_0}}-
\frac{H(X|YU)_{P_{X|U\neq u_0}}- H(X|Y)_{P_{\max}}}{H(X|U)_{P_{X|U\neq u_0}}- H(X)_{P_{\max}}}  
(H(X|U)_{P_{X|U\neq u_0}}-H(X|U)_P)
\nonumber \\
&\stackrel{(d)}{=}
\lambda' H(X|Y)_{P_{\max}}+(1-\lambda')H(X|YU)_{P_{X|U\neq u_0}},
\end{align}
where the steps $(a), (b), (c),$ and $(d)$ follow from 
\eqref{KG2}, \eqref{KG1}, \eqref{KG4}, and \eqref{KG3}, respectively.

This discussion shows that
a component $P_{X|U=u'}$ with $H(X)_{P_{X|U=u_0}}< H_0$
in the convex combination \eqref{KL1} can be replaced by $P_{\max}$.
Hence, we repeat the above procedure for all the elements $u \in {\cal U}$ such that
$P_U(u)>0$ and $H(X)_{P_{X|U=u}}\ge H_0$.
Then, we find another distribution $P'\in {\cal P}({\cal U}\times \cX)$ to satisfy the three required conditions.

\noindent {\bf Step 2}: Proof of (S2).

\noindent 

Next, we show (S2) by contradiction.
If there exists $u_1 \in {\cal U}$ such that $P_U'(u_1)>0$ and 
$H(X|Y)_{P_{X|U=u_1}'} >\kappa( H(X)_{P_{X|U=u_1}'} )$,
there exists another distribution $P''\in {\cal P}({\cal U'}\times \cX)  $
such that
$H(X|Y)_{P_{X|U=u_1}'} >H(X|YU')_{P''}$ and 
$H(X)_{P_{X|U=u_1}'} =H(X|U')_{P''}$.
Define the set ${\cal U''}:= ({\cal U}\setminus \{u_1\})\cup {\cal U'}$ and the distribution $P_{XU''}$
as
\begin{align}
P_{XU''}(x,u)=
\left\{
\begin{array}{ll}
P'(x,u) & \hbox{ when } u \in {\cal U} \\
P'_U(u_1)P''(x,u) & \hbox{ when } u \in {\cal U}' .
\end{array}
\right.
\end{align}
Hence, we have 
$H(X)_{P'}=\sum_{u\neq u_1 }P_U'(u)H(X)_{P_{X|U=u}'}+P_U'(u_1)H(X|U)_{P''}
=H(X|U'')_{P_{XU''}}$ and 
$H(X|Y)_{P'}>\sum_{u\neq u_1 }P_U'(u)H(X|Y)_{P_{X|U=u_1}'}+P_U'(u_1)H(X|YU)_{P''}=H(X|YU'')_{P_{XU''}}$.
Thus, we obtain the contradiction.

\noindent {\bf Step 3}: Proof of (S3).

\noindent 
We choose $P' \in {\cal P}({\cal U}\times {\cal X})$ given in (S1).
The statement (S2) implies the relations
$H(X|YU)_{P'}=\sum_{u \in {\cal U}}P_U'(u) H(X|Y)_{P_{X|U=u}'}$
and
$H(X|U)_{P'}=\sum_{u \in {\cal U}}P_U'(u) H(X)_{P_{X|U=u}'}$.
Thus, the point $(R_1,\kappa(R_1))=
 (H(X|U)_{P'},H(X|YU)_{P'})$ is written as
a convex combination of elements of \par
\noindent$\{(H(X)_P,H(X|Y)_P)\}_{P \in {\cal P}_0}$.
Hence, the desired statement is obtained.
\hspace*{\fill}~\QED\par\endtrivlist\unskip

\section{Proof of Lemma \ref{L4}}\Label{A1}
The relation \eqref{LLP2} follows from
\begin{align}
H(X^n) 
= \sum_{j=1}^n H( X_{j}|X^{j-1} )
\le \sum_{j=1}^n H( X_{j} ).\Label{LLP1}
\end{align}
Eq. \eqref{LLP3} follows from the relations;
\begin{align}
I(X^n; Y^n) = & H(Y^n)-H(Y^n|X^n) 
=H(Y^n)-\sum_{j=1}^n H(Y_j|X_j) \nonumber \\
\le & \sum_{j=1}^n H(Y_j)-H(Y_j|X_j) 
= \sum_{j=1}^n I(X_j; Y_j).
\end{align}
\hspace*{\fill}~\QED\par\endtrivlist\unskip

\bibliographystyle{IEEE}

\end{document}